\documentclass[a4paper, 11pt]{article}

\usepackage[papersize={210mm,297mm},top=30mm, bottom=30mm, left=26mm , right=26mm]{geometry}

\usepackage{lineno}
\modulolinenumbers[5]

\usepackage{listings}
\usepackage[usenames,dvipsnames]{xcolor}

\usepackage[colorlinks,citecolor=magenta,urlcolor=blue]{hyperref}

\usepackage{amsmath, amssymb}
\usepackage{amsthm}

\usepackage[linesnumbered,ruled]{algorithm2e}
\usepackage{booktabs}

\DeclareMathOperator*{\argmin}{argmin}

\newtheorem{proposition}{Proposition}[section]

\newtheorem{remark}{Remark}[section]

\usepackage{soul}
\usepackage{tikz}
\usetikzlibrary{arrows}
\usepackage{expl3}
\ExplSyntaxOn
\cs_set_eq:NN \fpeval \fp_eval:n
\ExplSyntaxOff

\usepackage{texments}

\begin{document}

\title{Change-in-Slope Optimal Partitioning Algorithm\\ in a Finite-Size Parameter Space}

\author{ Vincent Runge\\Universit\'e Paris-Saclay, CNRS, Univ Evry,\\ Laboratoire de Math\'ematiques et Mod\'elisation d'Evry,\\ 91037, Evry-Courcouronnes, France\footnote{E-mail: vincent.runge@univ-evry.fr}
  \and
Marco Pascucci\\
Universit\'e Paris-Saclay, CNRS, Univ Evry,\\ Institut des neurosciences Paris-Saclay, 91400, Orsay, France.
\and
Nicolas Deschamps de Boishebert\\Universit\'e Paris-Saclay, CNRS, Univ Evry,\\ Laboratoire de Math\'ematiques et Mod\'elisation d'Evry,\\ 91037, Evry-Courcouronnes, France
}

\date{}
\maketitle

\begin{abstract}
We consider the problem of detecting change-points in univariate time series by fitting a continuous piecewise linear signal using the residual sum of squares. Values of the inferred signal at slope breaks are restricted to a finite set of size~$m$. Using this finite parameter space, we build a dynamic programming algorithm with a controlled time complexity of~$O(m^2n^2)$ for~$n$ data points. Some accelerating strategies can be used to reduce the constant before~$n^2$. The adapted classic inequality-based pruning is outperformed by a simpler ``channel" method on simulations. Besides, our finite parameter space setting allows an easy introduction of constraints on the inferred signal. For example, imposing a minimal angle between consecutive segment slopes provides robustness to model misspecification and outliers. We test our algorithm with an isotonic constraint on an antibiogram image analysis problem for which algorithmic efficiency is a cornerstone in the emerging context of mobile-health and embedded medical devices. For this application, a finite state approach can be a valid compromise.
\end{abstract}


\section{Introduction}
\label{sec:intro}

Detecting change-points in time series is a long-standing problem in statistics that have been tackled in many different ways since the 1950s. Originally developed for quality control in manufacturing \cite{page1954continuous}, it has spread to the most modern sciences such as genomics \cite{hocking2015peakseg,cao2015changepoint}, neuroscience \cite{koepcke2016single,cribben2017estimating} or climate change \cite{beaulieu2012change,yu2019change,reeves2007review,cahill2015change} among many others. Due to the wide range of possible modeling assumptions and problem settings, change-point detection remains today an active scientific field, especially challenging for large data \cite{Frontiers2013}.

\subsection{Search Algorithms}
Last decades, researchers have been principally focused on the multiple change-point problem which consists in finding the location and the number of changes in a time series $y_{1:n} = (y_1,\dots,y_n)$ of length~$n$. Among the $2^{n-1}$ possible change-point vectors, the problem requires to select one vector for which a criterion is minimized or close to the minimum. Algorithmic efficiency is then a central challenge in order to rapidly infer a good change-point vector candidate. 

Binary segmentation (BS) is one of the best known algorithm \cite{Scott,Sen} for solving this problem. It returns an approximate solution in time $O(n\log(n))$. Many other approaches have been proposed: see \cite{truong2020selective} for a review and \cite{fearnhead2020relating} for some comparisons. A well-known alternative method is to use a dynamic programming approach \cite{bellman1962applied} called optimal partitioning (OP) \cite{jackson2005algorithm} which returns the best change-point vector minimizing the criterion\footnote{The problem can have many best change-point vectors with the same minimal value given a criterion. However, this multiple solution is an event of zero probability with real-valued time series.}. Its principle consists in finding the best last segment in consecutive truncated time series $y_{1:t}$ for $t$ from~$1$ to~$n$. Time complexity of this algorithm is of order $O(n^2)$ as for each $y_{1:t}$ we test the $t$ available positions (candidates) for potential best last change-point.

In recent years, pruning ideas were proposed to reduce the set of candidates to consider for each truncated time series $y_{1:t}$, so that the complexity of OP became comparable with BS in some simulation studies. Two main classes of pruning are available: inequality-based pruning implemented in the PELT algorithm \cite{Killick} and functional pruning \cite{rigaill2015pruned,Maidstone} implemented in the FPOP algorithm. PELT is based on a recursion on the position of the last change and remains close to OP (only a conditional expression is added). PELT is efficient only when the number of change-points is proportional to data length. FPOP is based on a recursion on the value of the last segment parameter and is more difficult to code: we need to update a functional cost at any step, that is a continuous piecewise quadratic function (in Gaussian model). However, this strategy achieves the best possible pruning and the log-linear time complexity on many data-sets~\cite{Maidstone}.

\subsection{Change in Slope With Continuity Constraint}

Many of these dynamic programming algorithms were developed to solve the change-in-mean problem for which we fit a piecewise constant signal. Another problem deals with changes in trend, also called changes in slope or changes in regression. It consists in fitting a piecewise linear signal to the time series. Many algorithms were developed for this task \cite{sharma2016trend} using for example Bayesian inference \cite{schutz2011detection} or Hypothesis testing \cite{militino2020performances} but most of them do not consider a continuity constraint between successive segments. This literature on change-point detection is often related to the estimation of climate variations \cite{beaulieu2012change}.

The change-in-slope problem with continuity constraint\footnote{In the following, we often write "change in slope" to mean "change in slope with continuity constraint".} is much less studied, even though many applications need such models \cite{yu2019change,jamali2015detecting,sowa2005direct,todd1999detecting}. A few approaches directly tackle this problem, using ideas of the Wild Binary Segmentation method \cite{fryzlewicz2014wild,baranowski2019narrowest}, trend filtering \cite{kim2009ell_1,tibshirani2014adaptive}, Bayesian approaches \cite{papastamoulis2020bayesian}, or dynamic programming \cite{fearnhead2018detecting}. The latter is based on an FPOP-like algorithm (CPOP) and fits a continuous piecewise linear signal based on the minimization of the residual sum of squares of each segment. It introduces a functional cost parametrized by the signal value associated to the current analysed data point. Using this cost, the authors derive a one-parametric functional update and eventually get a time complexity between $O(n^2)$ and $O(n^3)$ (see simulations in Section \ref{subsec:time}). 

\subsection{Contribution}

The guideline of our work was the construction of a simplified efficient algorithm for change-in-slope problem with continuity constraint, capable of returning the optimal change-point vector in a controlled time complexity of order $O(n^2)$. For this purpose, we use a dynamic programming approach with values of the inferred signal at slope breaks restricted to a finite set of values (called states). Furthermore, we were also interested in the possibility to easily constrain the inference as for change-in-mean algorithm GFPOP \cite{dylan2017log}. Our algorithm can restrict the inference to an increasing signal but also force a minimal angle value between consecutive segment slopes. This latter option provides stability for signal inference (filtering outliers) which is a desired features for many applications. The use of states is not only a bypass for reducing running time and enforce constraints: this restricted inference can match the expected level of complexity for some applications as for example in beat detection, for which the musician only needs integer beat-per-minute information \cite{hainsworth2004particle,robertson2012decoding}.

\subsection{Outline}

The outline of this paper is as follows. Section \ref{sec:OP} introduces the model and describes the optimal partitioning algorithm for changes in slope with a simple description of the continuity constraint. We also give the three proposed constrained inference modes: isotonic, unimodal and minimal angle.  Pruning-like approaches for speeding up the algorithm are exposed in Section \ref{sec:pruning}. Section \ref{sec:varest} is dedicated to variance estimation of time series with slopes. In the simulation study in Section \ref{sec:simu} we show the benefit of using the continuity hypothesis by measuring the mean squared error (MSE) and the Adjusted Rand Index (ARI) between the true signal and the inferred one. We also search for the range of penalties minimizing the MSE for 4 different simulation scenarios. We then compare the efficiency of the two proposed pruning methods and the time complexity with the FPOP-like challenging algorithm CPOP \cite{fearnhead2018detecting}. Eventually, we show how using the minimal-angle constraint enhances the stability of the inference. In last Section \ref{sec:app} we apply our algorithm to an antibiogram image analysis problem to find a unique change-point. In this application where non-decreasing signals are expected, the use of a monotonicity constraint improves detection precision. In this example, although the true signal is not expected to have a finite number of states, the finite-state trade-off allows a sensible gain in time and still a good analysis precision. In the context of mobile-health applications and hardware with limited computing capacity, where efficiency is important, such a trade-off could be fundamental.

Our change-in-slope algorithms are available into an R package on CRAN called slopeOP\footnote{\url{https://CRAN.R-project.org/package=slopeOP}} and on github\footnote{\url{https://github.com/vrunge/slopeOP}}.

\section{Change-in-slope Optimal Partitioning}
\label{sec:OP}

\subsection{Model and Cost Function}

We define the set of states $\mathcal{S}$ as a finite set of accessible real values for beginning and ending values in inferred segments. When we write $s= s_{min},\dots,s_{max}$, the variable $s$ goes through all the values of $\mathcal{S}$ from the smallest one to the biggest one. For computational efficiency we recommend to have $\# \mathcal{S} = m << n$ but this is not mandatory. The key idea is that limiting the number of states decreases computational time, while still capturing some useful information (see Section \ref{sec:app}). Notice that with too many states, the challenging algorithm CPOP would be more efficient.

Data are generated by the model \cite{fearnhead2018detecting}:
$$Y_t = s_{i} + \frac{s_{i+1}-s_{i}}{\tau_{i+1}-\tau_i}(t-\tau_i) + \epsilon_t\,,\quad t=\tau_{i}+1,\dots,\tau_{i+1}\,,\quad i = 0,\dots,k\,,$$
with $0 = \tau_0 < \tau_1 < \dots < \tau_k < \tau_{k+1} = n$, $s_0,\dots,s_{k+1} \in \mathcal{S}$ and $\epsilon_t \sim \mathcal{N}(0,\sigma^2)$ identically and independently distributed. Observations are denoted in lowercase $y_1,\ldots,y_n$. The change-point positions are given by integer vector $(\tau_0,\ldots,\tau_{k+1})$ and define a subdivision in $k+1$ consecutive segments $(y_{\tau_i+1},\dots,y_{\tau_{i+1}})$. This modeling imposes to define the cost for fitting data $y_{(\tau+1):t}$, $\tau < t$, with linear interpolation from value $s_1 \in \mathcal{S}$ to value $s_2 \in \mathcal{S}$ as a residual sum of squares:
\begin{equation}
\label{costLinear}
\mathcal{C}(y_{(\tau+1):t},s_1,s_2) = \sum_{i=\tau+1}^{t}\bigg(y_i - \Big(s_1 + (s_2-s_1)\frac{i-\tau}{t-\tau}\Big)\bigg)^2\,.
\end{equation}

Notice that the cost $(y_\tau-s_1)^2$ that would have been obtained at index $\tau$ is not present in the summation, unlike the right bound $(y_t-s_2)^2$. This choice of cost function will lead to an update rule with a simple description of the continuity condition for consecutive segments.

\subsection{Optimization Problem}

The {\it slopeOP} problem consists of finding the optimal partitioning of a time series $y_{1:n} = (y_1,\dots,y_n) \in \mathbb{R}^n$ of size $n$ into $k+1$ consecutive segments $(y_{\tau_i+1},\dots,y_{\tau_{i+1}})$ optimizing a penalized risk
\begin{equation}
\label{optimSlope}
Q_n^{slope} = \min_{\substack{\tau = (\tau_1,\dots,\tau_{k}) \in {\mathbb{N}}^{k} \\ \tau_0 = 0 \,,\, \tau_{k+1} = n\\ (s_0,\dots,s_{k+1}) \in \mathcal{S}}}\sum_{i=0}^k \big\{ \mathcal{C}(y_{(\tau_i+1):\tau_{i+1}}, s_i, s_{i+1}) + \beta\big\} - \beta\,,
\end{equation}
where the states defined inside the cost function provide the continuity constraint between successive segments and $\beta$ is a positive penalty. This penalty controls the amount of evidence we need to add a change: the greater this quantity, the less is $k$. The "minus beta" term is added in order to get the null value in case of a perfect fit (null residuals) with no change. We emphasise that $k$ is an unknown quantity. The penalty value is often set to the BIC penalty, $2\hat\sigma^2 \log(n)$  \cite{tickle2020parallelization,zheng2019consistency}, as we have a result of asymptotic consistency for change-in-mean problems  \cite{yao1988estimating}. The quantity $\hat\sigma^2$ is an estimation of the true variance~$\sigma^2$. Some other penalty values can be also relevant \cite{davis2006structural,zhang2007modified}. For the change-in-slope problem, theorems for a similar penalty have been given in asymptotic regime \cite{zheng2019consistency}. As variance estimation is needed, we proposed a robust estimator of the variance based on an adaptation to the difference estimator of Hall et al. \cite{hall1990asymptotically} (see {Table~\ref{tab:variance}}). In non-asymptotic regime, setting the penalty is a difficult task as it may also depends on signal shape (see Figure~\ref{scenarioResults}), that is the form of the deterministic part in the time series ($t \mapsto Y_t - \epsilon_t$).

\subsection{Dynamic Programming Algorithm}

To address the continuity constraint in an efficient algorithm we introduce the function 
$$Q_t : \left\{
\begin{array}{lll}
\mathcal{S} &\to &\mathbb{R} \\
u &\mapsto &Q_t(u)\,,
\end{array}
\right.$$
which is the optimal penalized cost up to position $t$ with the inferred value at this position equal to $u$. Our goal is then to update the set
$$\mathcal{Q}_t = \{Q_t(u), u= s_{min},\dots,s_{max}\}\,,$$
at any time step $t \in \{1,\dots,n\}$, where $s_{min}$ and $s_{max}$ are bounds that can be determined in a pre-processing step. With $y_{-}$ and $y_{+}$ the minimal and maximal values of the time-series, setting for example the bounds around $y_{\pm} \pm \frac{y_{+}-y_{-}}{2}$ include all extreme intersection points between consecutive segments (segment of two points $(y_{-},y_{+})$ followed by $(y_{+},y_{-})$ gives the upper bound). We find the objective value $\displaystyle Q_n^{slope}$ with the minimization $\displaystyle Q_n^{slope} = \min_{s_{min} \le u \le s_{max}}\big\{Q_{n}(u)\big\}$.

\begin{proposition}
\label{prop:updateRule}
The update rule with continuity constraint takes the form
\begin{equation}
\label{updateSlopeOP}
Q_t(v) = \min_{0 \le t' < t}\bigg( \min_{s_{min} \le u \le s_{max}}\big\{Q_{t'}(u) + \mathcal{C}(y_{(t'+1):t},u,v) + \beta\big\}\bigg)\,,
\end{equation}
where the state $u$ in $Q_{t'}$ and the cost function including states realizes the continuity constraint. At the initial step we have $Q_0(v) = -\beta$ for all $v \in \mathcal{S}$. 
\end{proposition}
The proof is detailed in Appendix~\ref{app:update} and the corresponding optimal partitioning algorithm based on (\ref{updateSlopeOP}) is given in Algorithm \ref{slopeOP}.

\begin{algorithm}[H]
\label{slopeOP}
\SetKwInOut{Input}{Input}
\Input{data $y_{1:n}$, set of states $\mathcal{S} = \{s_{min},\dots,s_{max}\}$ and penalty $\beta > 0$}
$Q$ = matrix of size $(n+1) \times m$\\
$cp$ = matrix of size $n \times m$\\
$U$ = matrix of size $n \times m$\\
\For{$v =  s_{min}$ to $s_{max}$}
{$Q(0,v) \leftarrow - \beta$}
\For{$t = 1$ to $n$}
{
\For{$v =  s_{min}$ to $s_{max}$}{$\displaystyle Q(t,v) \leftarrow \min_{0 \le t' < t}\left( \min_{s_{min} \le u \le s_{max}}\big\{Q_{t'}(u) + \mathcal{C}(y_{(t'+1):t},u,v) + \beta\big\}\right)$\\
$\displaystyle cp(t,v) \leftarrow \argmin_{0 \le t' < t}\left( \min_{s_{min} \le u \le s_{max}}\big\{Q_{t'}(u) + \mathcal{C}(y_{(t'+1):t},u,v) + \beta\big\}\right)$\\
$\displaystyle U(t,v) \leftarrow \argmin_{s_{min} \le u \le s_{max}} \left(\min_{0 \le t' < t} \big\{Q_{t'}(u) + \mathcal{C}(y_{(t'+1):t},u,v) + \beta\big\}\right)$}
}
return $cp$, $U$ and $Q(n,\cdot)$
\caption{Change-in-slope Optimal Partitioning (slopeOP)}
\end{algorithm}
~\\

The slopeOP algorithm is then an OP algorithm searching for the best last change-point in $y_{1:t}$ for $t$ from $1$ to $n$. But, differently from OP, it looks for it in a finite set of couples (state,time) of size $m \times t$, where the state is in the finite states set. Summing up for all states and for $t$ from $1$ to $n$, we get an overall running time of order $O(m^2 n^2)$. In next Section, we propose two accelerating strategies to reduce this initial $m^2 n^2$ time complexity. 

The backtracking step is a little different from the standard optimal partitioning algorithm as we need to take into account the states.\\

\begin{algorithm}[H]
\label{backtrackFSC}
\SetKwInOut{Input}{Input}
\Input{Output of Algorithm \ref{slopeOP} : $cp$, $U$ and $\{Q(n,v)\,,\, v \in \mathcal{S}\}$}
$chpts \leftarrow ()$\\
$\displaystyle states \leftarrow (\argmin_{s_{min} \le v \le s_{max}}\{Q(n,v)\})$\\
$t \leftarrow n$\\

\While{$t > 0$}{
$chpts \leftarrow  (t,chpts)$\\
$states \leftarrow (U(t,states(0)),states)$\\

$t \leftarrow cp(t,states(1))$\\

}
return $chpts$ and $states$
\caption{Backtracking Algorithm}
\end{algorithm}

\subsection{Inference With Signal Shape Constraints}
\label{subsec:const}

Besides its unambiguous computational time efficiency in $O(m^2n^2)$, one of the benefit of this finite-state approach is the possibility to simply enforce some constraints in the inference. 
The quantity to optimize is then
\begin{equation}
\label{optimSlopeConstr}
Q_n^{ctt} = \min_{\substack{\text{vectors } \tau \text{ and } s\\ (\tau,s) \in U_{k+1}}}\sum_{i=0}^k \big\{ \mathcal{C}(y_{(\tau_i+1):\tau_{i+1}}, s_i, s_{i+1}) + \beta\big\} - \beta\,,
\end{equation}
where $U_{k+1}$ is the set of constraints containing the positions (couples) that can be used in the minimization. We use notation $\nu_i$ for a couple $(\tau_i,s_i)$ position-state, $(\tau,s) = ((\tau',n),(s',v))$ is a sequence of couples with the last one being $(n,v)$ and $U_{k+1}^{(n,v)}$ is the set $U_{k+1}$ of sequences of couples with a last couple equal to $(n,v)$. A dynamic programming algorithm solving exactly (\ref{optimSlopeConstr}) can be built if in $U_{k+1}^{(n,v)}$ the information over the path $(\nu_1,\ldots,\nu_k)$ can be conveyed alongside this path to the current $\nu_k$ position. That is:
$$U_{k+1}^{(n,v)} = \{(\tau,s)\,|\, (\tau',s') \in U_k^{\nu_k}\,,\, f(\nu_1,\ldots,\nu_k,(n,v)) = 1  \}$$
$$ = \{(\tau,s)\,|\, (\tau',s') \in U_k^{\nu_k}\,,\, g(M(\nu_k),\nu_k,(n,v)) = 1  \}\,,$$
with $g:\mathbb{R}\times (\{0,\dots,n\} \times \mathcal{S})^2 \to \{0,1\}$ a validity test associated to the constraint. Function $f$ is reduced to $g$. The ``memory'' function $M:\{0,\dots,n\} \times \mathcal{S} \to \mathbb{R} $ summarizes the information of the path $(\nu_1,\ldots,\nu_k)$ and is associated in practice to the cost $Q_{t'}(u)$ of the update rule. If $g$ equals $0$, the couple $\nu_k = (t',u)$ cannot be considered in the minimization. 

With this definition for $U_{k+1}$ we get the update rule:
\begin{equation}
\label{updateSlopeOP_constraint}
Q_t^{ctt}(v) = \min_{\substack{ \tilde{\nu} = (t',u) \in  \{0,\dots,t-1\} \times \mathcal{S}\\ \nu = (t,v) \,,\, g(M(\tilde{\nu}),\tilde{\nu},\nu) = 1}}\bigg( Q_{t'}(u) + \mathcal{C}(y_{(t'+1):t},u,v) + \beta\bigg)\,.
\end{equation}
In this rule, the positions $\tilde{\nu}$ of the minimization operator are taken into account according to the $g$ function value depending only on the current couple $\nu = (t,v)$ and past events memorized by function $M(\cdot)$. We illustrate this constraint approach by the examples of isotonic, unimodal and stable inferences.

{\bf Isotonic case.} The update rule for isotonic constraint is:
\begin{equation}
\label{updateSlopeOP_isotonic}
Q_t^{iso}(v) = \min_{0 \le t' < t}\bigg( \min_{s_{min} \le u \le v}\big\{Q_{t'}(u) + \mathcal{C}(y_{(t'+1):t},u,v) + \beta\big\}\bigg)\,,
\end{equation}
which is close to (\ref{updateSlopeOP}) but with a constrained minimization for variable $u$. Here $M(\tilde{\nu})=0$ for all $\tilde{\nu}$ and $g(0,\tilde{\nu},\nu) = 1$ if and only if $u \le v$ for  $\tilde{\nu} = (t',u)$ and $\nu=(t,v)$. We illustrate the use of the isotonic constraint in Section \ref{sec:app} with an antibiogram image analysis problem.\\

{\bf Unimodal case.} $M(\tilde{\nu})$ is equal to $1$ if at position $\tilde{\nu} = (t',u)$ no decreasing segment has been yet inferred. Otherwise $M(\tilde{\nu})=0$. $g(1,\cdot,\cdot) = 1$ and 
$$g(0,\tilde{\nu} = (t',u),\nu = (t,v)) = \left\{
\begin{array}{ll}
0 & \hbox{if} \,\,u < v\,,  \\
1 & \hbox{if} \,\, u \ge v \,.\\
\end{array} \right.$$
Value $M(\nu)$ is determined just after setting $Q(\nu) = Q_{t}(v)$ using the slope sign information on the segment $(\tilde{\nu},\nu)$.\\

{\bf Minimal angle case.} $M(\tilde{\nu}) \in \mathbb{R}$ is the slope value of the last inferred segment at position $\tilde{\nu}$. Here the value $g(M(\tilde{\nu}),\tilde{\nu},\nu)$ computes the angle deviation in degree between slope $M(\tilde{\nu})$ ans slope formed by $(\tilde{\nu},\nu)$ If the obtained value is less that a threshold we have $g(M(\tilde{\nu}),\tilde{\nu},\nu) = 0$ and $g(M(\tilde{\nu}),\tilde{\nu},\nu) = 1$ otherwise. With this constraint, we tend to be robust to outliers with the right level of smoothness (the threshold). We also hope to get some kind of robustness to model misspecification (see simulations in Section \ref{subsec:minangle}).\\

In the R package slopeOP, we also provide a \{Segment Neighborhood dynamic programming method returning the best change-point vector with a given fixed number of segments. With $k$ changes, its time complexity is about $k$ times greater than the complexity of slopeOP.

\section{Accelerating Strategies}
\label{sec:pruning}

As revealed by Algorithm \ref{slopeOP} the double loop for variables $t$ and $s$ (lines 7~and~8) is time consuming. We first give a formula for a constant-time computation of cost (\ref{costLinear}). Then we propose two accelerating strategies aiming at reducing the set of values to consider for the search for minimum in matrix $Q$. The first one is a PELT-like method based on inequalities of type ``$Q_{t'}(u) + \mathcal{C}(y_{{(t'+1)}:t},u,v) > Q_t(v)$". The second one considers direct comparisons between elements of the minimization in (\ref{updateSlopeOP}). We test the efficiency of those two accelerating rules on simulations in Section \ref{subsec:time}. As the constraints can force the choice of non-minimal elements in (\ref{updateSlopeOP_constraint}) these methods can only be implemented in non-constraint setting (at the exception of the isotonic constraint).

\subsection{Efficient Cost Computation}

\begin{proposition}
\label{efficientCost}
The cost function (\ref{costLinear}) can be computed in constant time with the formula 
$$\mathcal{C}(y_{(t'+1):t},u,v) = S_{t}^2 - S_{t'}^2 - \frac{2}{t-t'} \Big( \big(ut-vt'\big)\big(S_{t}^1 - S_{t'}^1\big) + (v-u)\big(S_{t}^+ - S_{t'}^+\big)\Big)$$
$$\quad\quad\quad\quad\quad\quad + \frac{v^2-u^2}{2} + \frac{u^2 + uv +v^2}{3}(t-t') + \frac{(v-u)^2}{6(t-t')}\,,$$
where $t' < t \in \{0,\dots,n\}$, $(u,v) \in \mathcal{S}^2$ and
$$S^1_t = \sum_{i=1}^t y_i\quad , \quad S^2_t = \sum_{i=1}^t y_i^2\quad \hbox{and} \quad S^+_t = \sum_{i=1}^t iy_i\quad \hbox{for all} \,\, t \in \{1,\dots,n\}\,.$$
\end{proposition}

The proof is straightforward by direct (tedious) calculations. The vectors $(S_t^1)_{t=1,\dots,n}$, $(S_t^2)_{t=1,\dots,n}$ and $(S_t^+)_{t=1,\dots,n}$ are computed in a pre-processing step in $O(n)$ time and saved. This is a marginal cost compared to the cost for Algorithm \ref{slopeOP}. Notice that the  $O(m^2n^2)$ time complexity is obtained using the result of Proposition \ref{efficientCost}.

\subsection{Inequality-based Pruning}

With the following Proposition \ref{P3} we build a standard inequality-based pruning rule that has the specificity to take into account future data. Our cost function~(\ref{costLinear}) doesn't use a minimisation as in PELT: we cannot guarantee that splitting a segment into two segments would result in a smaller overall cost. This property is at the core of the PELT pruning rule. 

At each position $v \in \mathcal{S}$, we need to update a set of positions:
$$\mathcal{U}(t,v) \subset \{(t',u) \,,\, t' \in \{0,\dots,t-1\}\,,\, u \in \{s_{min},\dots,s_{max}\}\}\,,$$
to get $\mathcal{U}(t+1,v)$. There are $m$ such sets. Therefore, we actually chose to transfer information at constant state value.

\begin{proposition}
\label{P3}
At time $t$ and position $v$, suppose that there exists $t'<t$ and $u\ne v$ such that 
\begin{equation}
\label{Rule3Pruning}
Q_{t'}(u) + \mathcal{C}(y_{{(t'+1)}:t},u,v) > Q_t(v)\,.
\end{equation}
We define  $S_t^{t'} = \sum_{i=t'+1}^{t}(i-t')y_i$ and coefficients $(\alpha_t^-,\alpha_t^+,\gamma^-_t,\gamma^+_t)$ such that
$$\alpha_t^+ T + \gamma^+_t \le \sum_{i=t+1}^{T-1}y_i \frac{T-i}{T-t} \le \alpha_t^- T + \gamma^-_t\,, \quad T= t+1,\dots,n\,.$$
We also define for all integers $T > t$ the affine in $T$ functions: 
$$ \left\{
\begin{array}{l}
\displaystyle f^+(T) = \left(\alpha^+_t - \frac{u + 2v}{6}\right)T + t'\frac{u + 2v}{6}  - \frac{v-u}{12(t-t')} +  \frac{S_t^{t'}}{t-t'} + \gamma^+_t\,, \\
\displaystyle f^-(T) = \left(\alpha^-_t - \frac{u + 2v}{6}\right)T + t'\frac{u + 2v}{6} - \frac{v-u}{12(t-t')} +  \frac{S_t^{t'}}{t-t'} + \gamma^-_t\,.
\end{array}
\right.
$$

If $f^+(t+1) \ge 0$ and $f^+(n) \ge 0$ (case $v - u > 0$) or if $f^-(t+1) \le 0$ and $f^-(n) \le 0$ (case $v - u < 0$) then the position $(t',u)$ doesn't have to be considered for further iterations $T>t$ for computing $Q_{T}(v)$ in Algorithm~\ref{slopeOP}. We remove $(t',u)$ from $\mathcal{U}(T,v)$ for all $T>t$.
\end{proposition}

The proof is given in Appendix~\ref{app:pelt}.

The use of additional parameters (the states) in cost function leads to a more complex pruning rule than the one for standard PELT \cite{Killick} as the overall cost is not guarantee to decrease as the number of change-point increases. However, in case $u=v$ we have a simpler result:

\begin{proposition}
\label{P1}
At time $t$, if there exists $t' \in \{0,\dots,t-1\}$ such that 
\begin{equation}
\label{Rule1Pruning}
Q_{t'}(v) + \mathcal{C}(y_{{(t'+1)}:t},v,v) > Q_t(v)\,,
\end{equation}
then the position $t'$ doesn't have to be considered for further iterations $T>t$ for computing $Q_{T}(v)$ in Algorithm \ref{slopeOP}.
\end{proposition}

\begin{proof}
We use the fact that we have
$$ \mathcal{C}(y_{{(t'+1)}:t},v,v) +  \mathcal{C}(y_{(t+1):T},v,v) =  \mathcal{C}(y_{{(t'+1)}:T},v,v)\,,$$
for any $t' < t < T$. Adding the quantity $\mathcal{C}(y_{(t+1):T},v,v)$ to (\ref{Rule1Pruning}) and using this equality we get
$$Q_{t'}(v) + \mathcal{C}(y_{{(t'+1)}:T},v,v) + \beta > Q_t(v) + \mathcal{C}(y_{(t+1):T},v,v) + \beta\,,$$
which means that $(t',v)$ is a choice less optimal than $(t,v)$ for the minimization in $Q_{T}(v)$.
\end{proof}

We have a second pruning rule for $v$-constant positions which is original in the sense that a recent position can be pruned by an older one.

\begin{proposition}
\label{P2}
At time $t$, if there exists $t' \in \{0,\dots,t-1\}$ such that 
\begin{equation}
\label{Rule2Pruning}
Q_{t'}(v) + \mathcal{C}(y_{{(t'+1)}:t},v,v) \le Q_t(v)\,,
\end{equation}
then the position $t$ doesn't have to be considered for further iterations $T>t$ for computing $Q_{T}(v)$ in Algorithm \ref{slopeOP}.
\end{proposition}

The proof is the same as for Proposition \ref{P1}.

\begin{remark}
One of the two inequalities (\ref{Rule1Pruning}) and (\ref{Rule2Pruning}) is always satisfied so that the subset $\mathcal{U}(t,v) \cap \{(t',v)\,,\,t'=1,\dots,t-1\}$ contains at most one element.
\end{remark}

\subsection{Channel Pruning}

We propose a new rule for speeding-up the algorithm, called the ``channel method". It does not prune the position $(t',u)$ in the classical sense, as a non-considered value can be further reintegrated in the minimization (in a set $\mathcal{U}(t,v)$). However its simplicity could help to reduce time complexity.

We expect the cost $Q$ to verify some structural properties. In particular, 
$$q_t = \left\{
\begin{array}{lll}
\mathcal{S} &\to &\mathbb{R} \\
v &\mapsto &Q_t(v)
\end{array}
\right.$$ should have a minimum for $v$ ``near" the data, which means that $q_t$ is, for most $t$, decreasing until a state $v_l$ and increasing from a state $v_u \ge v_l$. The function $$C_{(t'+1):t}^{\tilde{v}} = \left\{
\begin{array}{lll}
\mathcal{S} &\to &\mathbb{R} \\
v &\mapsto & \mathcal{C}(y_{(t'+1):t},v,\tilde{v})
\end{array}
\right.$$
has a global minimum that can be easily found (it's a quadratic function). The idea is to study the variations of $q_{t'}+C_{(t'+1):t}^{\tilde{v}}$ in order to leave out the state values for which we know they are away from the argminimum.

\begin{proposition}
\label{P4}
We consider the function $q_{t'}+C^{\tilde{v}}_{(t'+1):t}$ with $t'+1 < t$. The minimum of this function is inside the interval
 $$I_{t'}^{\tilde{v}} = \mathcal{S} \setminus (S_{min} \cup S_{max})\,,$$
with
$$ \left\{
\begin{array}{l}
S_{min} = \{s_{min},\dots,\min(v_l^{-1},[v^*]^{-1})\}\,, \\
S_{max} = \{\max(v_u^{+1},[v^*]^{+1}),\dots,s_{max}\}\,,
\end{array}
\right.
$$
 where $[\cdot]$ is the nearest ``value in set $\mathcal{S}$" operator and $(\cdot)^{+1}$ (resp. $(\cdot)^{-1}$) denotes the following (resp. preceding) value in the ordered set $\mathcal{S}$. Function $q_{t'}$ is decreasing on $\{s_{min},\dots,v_l\}$ and increasing on $\{v_u,\dots,s_{max}\}$. We also know that the argminimum of $C^{\tilde{v}}_{(t'+1):t}$ is given by formula
$$v^* = \frac{6}{(t-t'-1)(2(t-t')-1)}\sum_{i = t'+1}^{t}\Big[(t-i)y_i\Big] - \tilde{v} \frac{t-t'+1}{2(t-t')-1}\,.$$
\end{proposition}

\begin{proof}
$v^*$ is easy to compute as the argminimum of the quadratic function $C^{\tilde{v}}_{(t'+1):t}$. The set $S_{min}$ (resp. $S_{max}$) corresponds to the set on which we know that function $q_{t'}+C^{\tilde{v}}_{(t'+1):t}$ is decreasing (resp. increasing) with the next (resp. previous) value in $\mathcal{S}$ giving a smaller output.
\end{proof}

Notice that the closest state search $[v^*]$ is a $O(1)$ operation if $\mathcal{S}$ is made of consecutive integers or a $O(\log(n))$ operation in a general ordered list. We define the matrix $\mathcal{Q} \in \mathbb{R}^{m \times n}$ with $\mathcal{Q}_{ij} = Q_j(s_i)$. The name ``channel" comes from the fact that the matrix $\mathcal{Q}$ is restricted by a channel with fixed $v_u$ and $v_l$ states for each column of this matrix. For all $\mathcal{U}(t,v)$, this channel is updated to $I_{t'}^{\tilde{v}}$ after computing the value $v^*$ and using Proposition \ref{P4}.

\section{Variance Estimation}
\label{sec:varest}

With real data-sets the noise level $\sigma^2$ in time series is an unknown quantity that has to be estimated. In change-in-mean problems there exist many estimators for the variance of a time-series as for example the mean absolute deviation (MAD) or HALL estimators \cite{hall1990asymptotically}. They have the property to be slightly sensitive to change-points when the number of changes is small compared to data length. In slope problems we face the additional difficulty to remove the slope effect. We suggest to adapt the HALL estimator to our change-in-slope problem. In classic mean problem, we have for HALL of order $3$: 
$$\hat{\sigma}^2_{mean} = \frac{1}{n-3}\sum_{j=1}^{n-3}\Big( \sum_{k=0}^3 d_k Y_{j+k}\Big)^2\,,$$
with $d_0 = 0.1942$, $d_1 = 0.2809$, $d_2 = 0.3832$ and $d_3 = -0.8582$. We apply Hall on the successive differences (``HALL Diff"):
$$\hat{\sigma}^2_{slope} = \frac{1}{(n-4)\Delta}\sum_{j=1}^{n-4}\Big( \sum_{k=0}^3 d_k (Y_{j+k+1}-Y_{j+k})\Big)^2\,,$$
with $\Delta = d_0^2 + (d_1-d_0)^2 + (d_2-d_1)^2 + (d_3-d_2)^2 + d_3^2 = 1.527507$ a normalization coefficient.

To evaluate the quality of the HALL Diff estimator, we simulate time series of length $100$ with $2$ different signals: the first one is piecewise linear and the second one is sinusoidal, as shown in Figure \ref{fig:sigmaEstim}.

\begin{figure}[!h]
\center\includegraphics[width=0.8\textwidth]{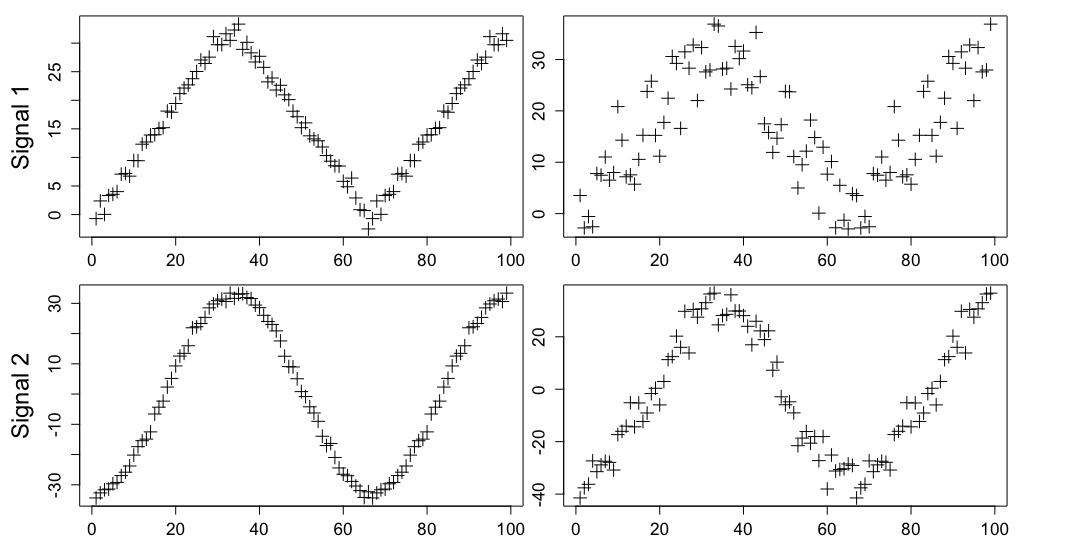} 
\caption{Examples for the linear and sinusoidal time series with noise $\sigma = 1$ (left) and $\sigma = 5$ (right).}
\label{fig:sigmaEstim}
\end{figure}

In Table \ref{tab:variance}, we present results for estimating the standard deviation $\sigma$ for $10^4$ simulated time series for each level of noise ($\sigma= 1,\dots,5$). We easily conclude that the HALL Diff method gives very accurate results, even with a non-linear (sinusoidal) model. We also have a better precision of the estimated sigma for HALL Diff than for MAD. Simulations with more change-points (10) yield similar results (see Appendix~\ref{app:variance}).

\begin{table}[!h]
\begin{tabular}{c|c|cccccc}
method&signal& $\sigma$ &1& 2 & 3 & 4 & 5 \\
\hline\hline
MAD&slope&mean($\hat\sigma$)& 1.23 & 2.11 & 3.08 & 4.05 & 5.05  \\
& &sd($\hat\sigma$)& 0.13 & 0.25 & 0.39 & 0.52 & 0.64\\
\cline{2-8}
&sinus&mean($\hat\sigma$)& 1.93 & 2.54 & 3.36 & 4.26 & 5.22  \\
&&sd($\hat\sigma$)& 0.16 & 0.26 & 0.38 & 0.51 & 0.64\\
\hline\hline
HALL&slope&mean($\hat\sigma$)&  1.84 & 2.52 & 3.37 & 4.28 & 5.23 \\
 &&sd($\hat\sigma$)& \it 0.048 & \it 0.13 & \it 0.21 & \it 0.29 &\it  0.37 \\
\cline{2-8}
&sinus&mean($\hat\sigma$)&  3.58 & 3.97 & 4.56 & 5.27 & 6.06 \\
&&sd($\hat\sigma$)& \it 0.026 & \it 0.082 & \it 0.16 & \it 0.24 & \it 0.32\\
\hline\hline
HALL Diff&slope&mean($\hat\sigma$)& \bf 1.01 & \bf 2.00 & \bf 3.00 & \bf 3.99 & \bf 4.99  \\
 &&sd($\hat\sigma$)& 0.087 & 0.18 & 0.27 & 0.36 & 0.44\\
\cline{2-8}
&sinus&mean($\hat\sigma$)& \bf 1.02 & \bf 2.00 & \bf 3.00 & \bf 3.99 & \bf 4.99 \\
&&sd($\hat\sigma$)& 0.087 & 0.18 & 0.26 & 0.36 & 0.44  \\
\hline\hline
\end{tabular}
\caption{Variance estimation with MAD, HALL and HALL Diff estimators. The closest values to the true sigma are in bold; the smallest standard deviations in italic.}
\label{tab:variance}
\end{table}

\section{Simulation Study}
\label{sec:simu}

Our simulation study is split into three parts. We first compare the mean squared error (MSE) and the Adjusted Rand Index (ARI) between the true signal and the inferred one with many different algorithms and study the impact of choosing different $\beta$ penalty values. In the next part, we compare the two proposed pruning rules and also select our main competitor, CPOP, to challenge slopeOP in terms of computational efficiency. Finally, we consider misspecified time series with an heavy-tailed noise and explore the capacity of slopeOP with a minimal angle between segments to infer a good model on a range of penalties.

\subsection{MSE Competition with Other Algorithms}

We consider the following five methods:
\begin{itemize}
    \item OP-2D : we fit the data with the PELT algorithm but the cost function on a segment is the residual sum of squares between data and a linear regression (2-dimensional fit). In that case, we don't have any continuity constraint;
    \item FPOP : The standard efficient gaussian FPOP algorithm \cite{Maidstone} used on differenced data $z_t = y_{t+1}-y_t$, $t=1,\dots,n-1$. The continuity is obtained by construction;
    \item RFPOP : Same algorithm as the previous FPOP but with robust biweight loss (with robust parameter $K = 3\sigma$) \cite{Fearnhead};
    \item CPOP : A FPOP-like algorithm for changes in slope with continuity constraint \cite{fearnhead2018detecting}. The set of states is here infinite-dimensional ($\mathbb{R}$) and not finite as for slopeOP; 
     \item slopeOP : Our finite-state change-in-slope OP algorithm.
\end{itemize}

\begin{figure}[!h]
\center
\includegraphics[width=1\textwidth]{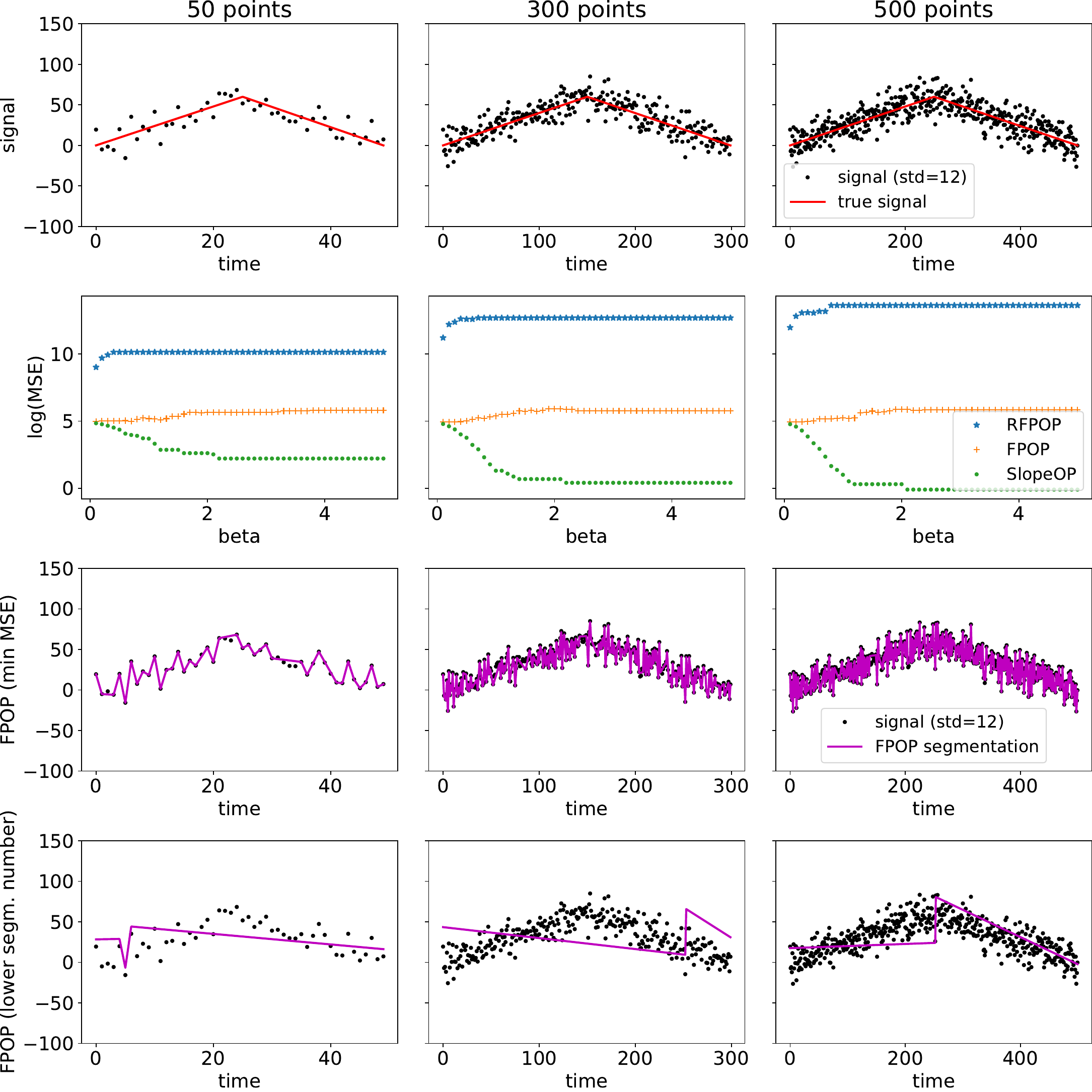} 
\caption{The first row presents 3 time-series of size $50$, $300$ and $500$ that are obtained by the black hat-shaped signal with an additive Gaussian noise with $\sigma = 12$. For a range of values for $b = \frac{\beta}{\sigma^2\log(n)}$ we plot the average log(MSE) curve for algorithms FPOP, RFPOP and slopeOP for $10$ simulated time series. The first two algorithms give performances far away from slopeOP. In the third row we show an example of an inferred signal by FPOP with a $b$ value corresponding to the minimum in log(MSE) curves. In last row we plot the segments obtained by the biggest possible $b$ value before having a unique segment with the same data. We observe that we fail to obtain only $2$ segments and that we obviously obtain an unrealistic result.}
\label{fpoprfpop}
\end{figure}

The first idea to detect change in slope consists of looking at changes in data $z_t= y_{t+1}-y_t$ with any well-known change-in-mean algorithm. We use here the efficient FPOP and robust FPOP algorithms for the inference. We highlight in Figure~\ref{fpoprfpop} the fact that with a simple hat-shaped model and a low level of noise, the FPOP algorithms give bad performances as the number of datapoint increases. Therefore, we left out theses two methods and study only the behavior of three algorithms: OP-2D, CPOP and slopeOP.

\begin{figure}[!h]
\center
\includegraphics[width=\textwidth]{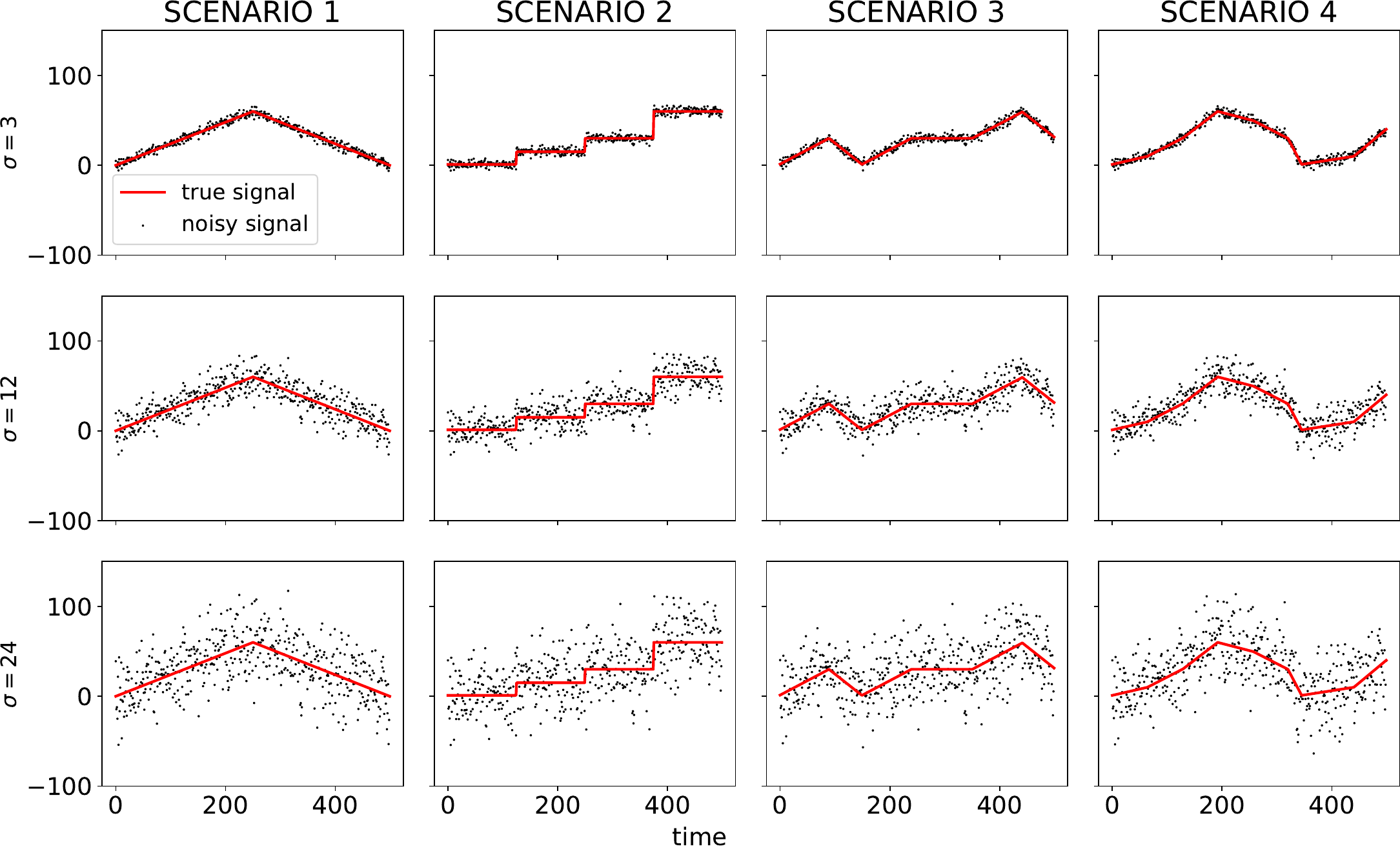} 
\caption{The plain black line is the true signal. For each scenario we simulate time series with an additive Gaussian noise of increasing standard deviation ($\sigma = 3, \sigma = 12, \sigma = 24$). Each plot shows a realisation of this procedure with $500$ data points. The number of segments for scenarios $1$ to $4$ is respectively $2$, $7$, $6$ and $8$.}
\label{4scenarios}
\end{figure}

We now consider 4 scenarios of increasing complexity with 3 levels of noise as shown in Figure~\ref{4scenarios}. The beginning and ending values of all segments are integers in $\{0,\dots,60\}$ and we choose a set of states $\mathcal{S} = \{-10,\dots,70\}$ to reduce side effects in limit points ($0$ and $60$).\\
We simulate $150$ time series of length $500$ for each of the couple (scenario, noise) and draw in Figure~\ref{scenarioResults} the graph of the MSE, the ARI and computation time with respect to $b = \frac{\beta}{\sigma^2\log(n)}$, where $\sigma$ is the chosen known level of noise. For each of the $150$ time series, the $3$ algorithms are run for $b = 0.1,0.2,\dots,4.9,5$ and each point of the curves (log(MSE), ARI or time) is the mean over the $150$ independent simulations. 

We discuss the results for the intermediate level of noise ($\sigma = 12$) exposed in Figure~\ref{scenarioResults}, other results are exposed in Appendix~\ref{app:scenarios}. For all scenarios CPOP and slopeOP give very close and similar good results compared to OP-2D. The lack of continuity constraint explains the higher MSE value and lower ARI for OP-2D in scenarios $1$, $3$ and $4$. For all scenarios and all noise levels, the optimal penalty in our simulations with $n=500$ is close to $b = 2$ for slopeOP and CPOP as expected in asymptotic regime \cite{zheng2019consistency}. Execution time is smaller for CPOP in most scenarios as data length is here limited to $500$ (see next Subsection for more details). In Figure~\ref{bestMSEsegmentation}, we plotted the segmentation obtained by OP-2D and slopeOP at the smallest MSE value. The result highlight the need for a continuity constraint to get a better inference.


\begin{figure}[!h]
\center
\includegraphics[width=\textwidth]{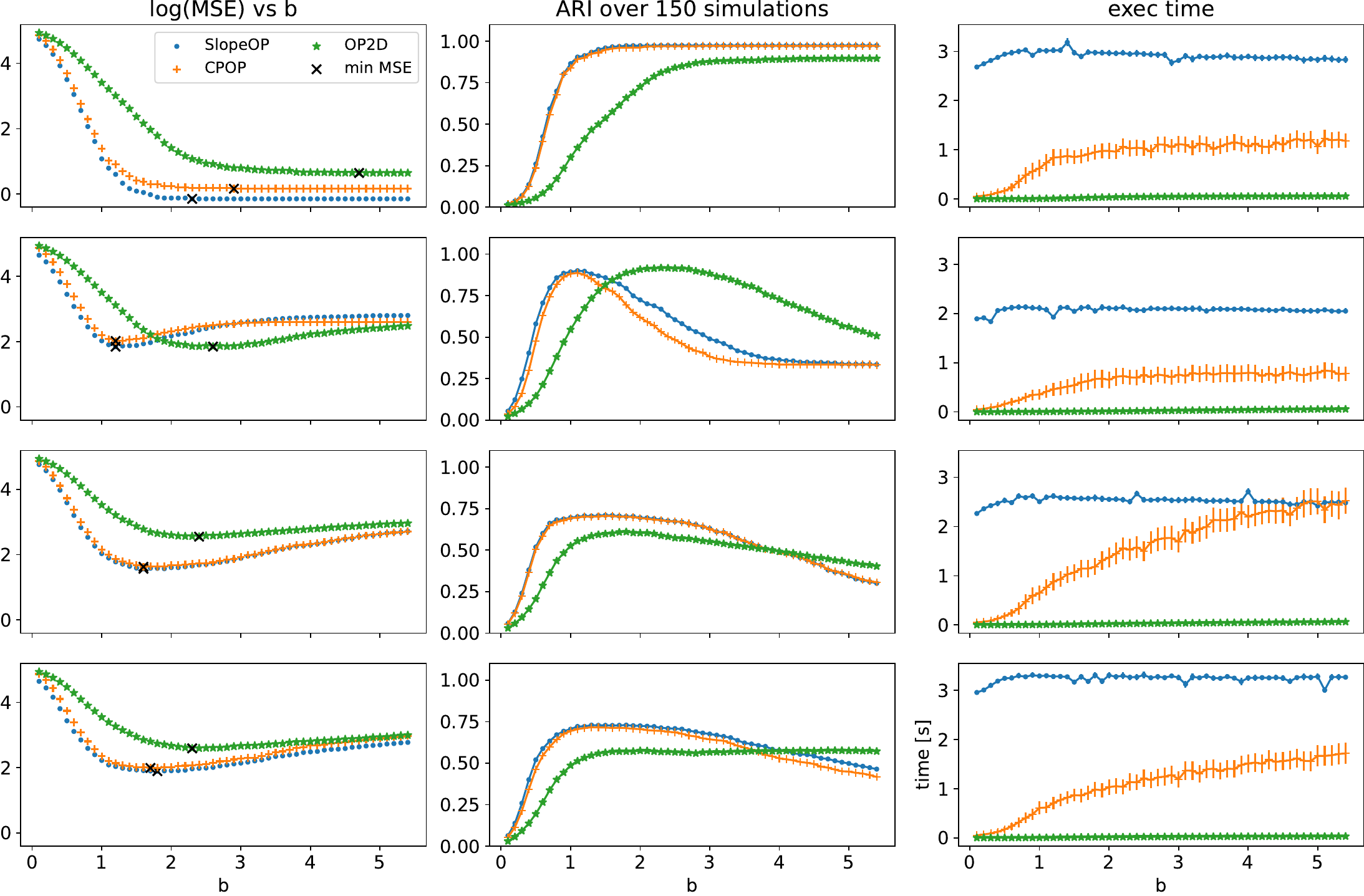} 
\caption{Results for scenarios $1$ to $4$ with noise $\sigma = 12$. Each point of the MSE curve is the mean over $150$ independent simulations.
Columns 2 shows the Adjusted Rand Index (ARI) of the clustering obtained with SlopeOP and CPOP compared with the true signal.
The last column shows the execution time vs $b$.
The results of SlopeOP and CPOP are very similar. SlopeOP is slower than CPOP in most cases, but it's complexity does not depend on $b$.
}
\label{scenarioResults}
\end{figure}

\begin{figure}[!h]
\center
\includegraphics[width=\textwidth]{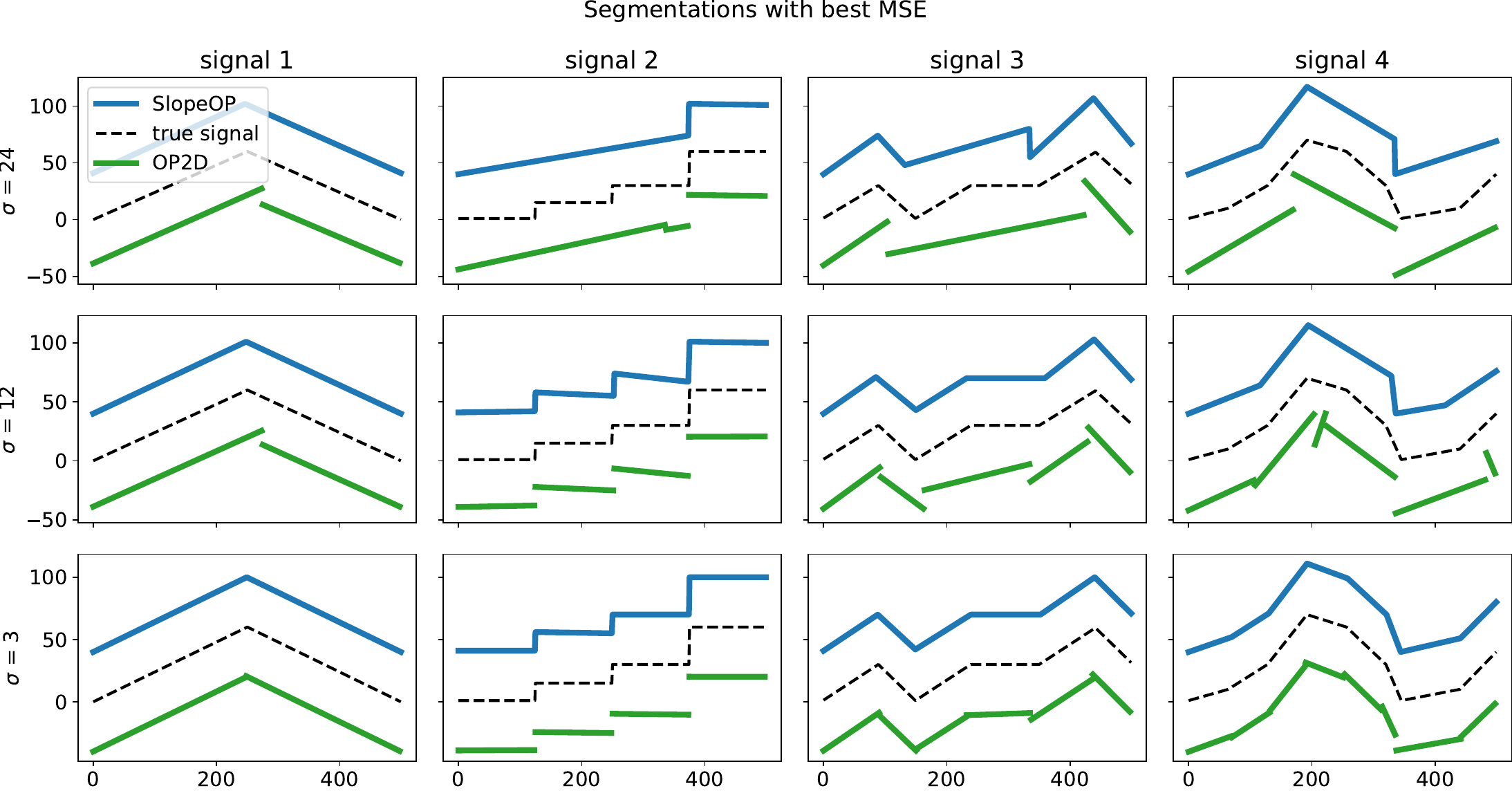} 
\caption{Segmentation obtained with values of $b$ that minimise the MSE. These values correspond to the black crosses in figure \ref{scenarioResults}.
The curve are deliberately y-shifted for better appreciation of the shape differences.}
\label{bestMSEsegmentation}
\end{figure}

\subsection{Computational Efficiency}
\label{subsec:time}

We first compare time efficiency for the two accelerating methods: channel-based rule and inequality-based pruning. We consider hat-shaped time series with a signal from $10$ to $50$ and $61$ integers states ($\mathcal{S} = \{0,\dots,60\}$) for inference. In Figure \ref{fig:timePruning}, we ran $100$ repetitions for each sigma value with data of length $500$ and two hat-shaped signals (with $1$ hat or $25$ regularly spaced hats). We easily understand that the channel rule is more efficient for all noise levels. It may be possible to get similar pruning efficiency in some very particular situations but we conclude with these simulations that the default rule should always be the channel method. Morever, this latter is simpler as it processes each element in the cost matrix quicker (compare Propositions~\ref{P3} and~\ref{P4}).

\begin{figure}[!h]
\center
\includegraphics[width=0.98\textwidth]{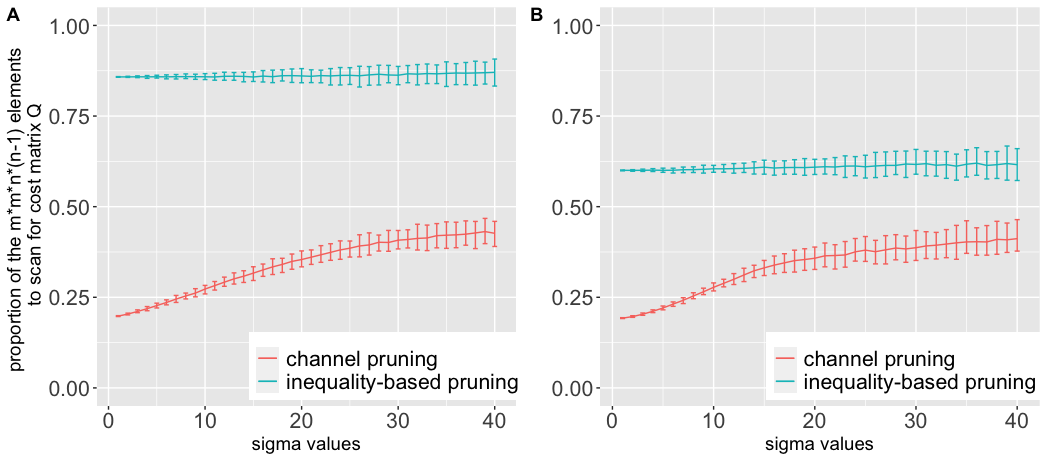} 
\caption{Without pruning the dynamic programming algorithm needs to scan $m^2n(n-1)$ elements for finding mimima at each step. We compare the proportion of elements that the algorithm scans for channel-based and inequality-based prunings with different noise levels. In panel A we use a hat-shaped signal of length $500$ with $1$ hat and a signal with $25$ regular hats in panel B. For each noise level, we simulate $100$ time series and plot the error bar using the $95\%$ empirical confidence interval. For all noise levels, the channel method is more efficient.}
\label{fig:timePruning}
\end{figure}

We now compare time efficiency between slopeOP and CPOP algorithms with two noise levels and two different hat-shaped signals over $100$ simulations for each data length regularly sampled between $n=100$ and $n=1500$ on the log scale. In Figure \ref{fig:time}, we plot the mean time for each algorithm with a $1$-hat signal and two levels of noise in log-log scale. The results confirm the quadratic complexity for slopeOP and a complexity for CPOP between $n^2$ and $n^3$ (closer to $n^3$ with a higher level of noise). In this simulation, slopeOP outperforms CPOP for time series of length greater than $550$ with $\sigma = 3$ and $1000$ with $\sigma = 24$ (for example, with $n=1500$ and $\sigma = 3$ we have a mean time of $32s$ for CPOP and $10s$ for slopeOP).  The coefficient $q$ in complexity $O(n^q)$ is equal to $2.88$ and $2.94$ for CPOP for $\sigma = 3$ and $\sigma = 24$, respectively. For slopeOP we get $1.95$ and $2.04$ with the channel pruning option. With a signal of $10$ regular hats, CPOP is faster up to $n=1500$ with power coefficient for CPOP $q_{\sigma = 3}=2.16$ \& $q_{\sigma = 24}=2.47$ and for slopeOP $q_{\sigma = 3}=1.96$ \& $q_{\sigma = 24}=2.03$. The corresponding plot is presented in Appendix~\ref{app:time}.

\begin{figure}[!h]
\center
\includegraphics[width=1\textwidth]{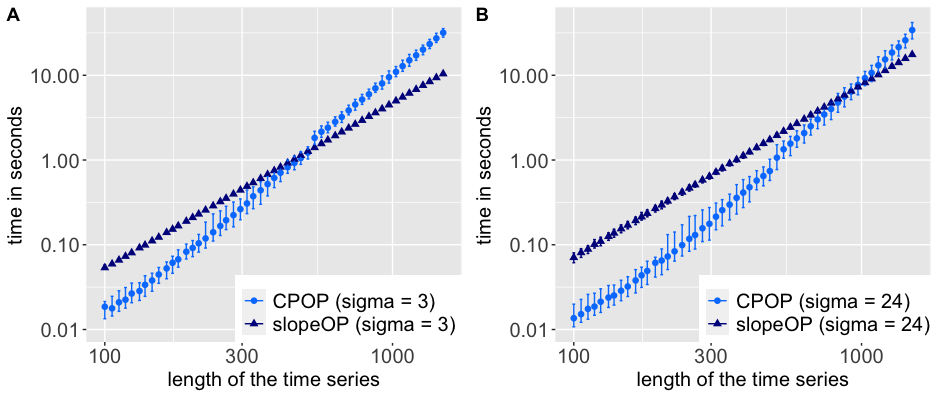} 
\caption{Computational time of CPOP versus slopeOP for one-hat-shaped data with noise level $\sigma =3$ and $\sigma =24$ in log-log scale. The coefficient $q$ in complexity $O(n^q)$ is equal to $2.88$ and $2.94$ for CPOP for $\sigma = 3$ and $\sigma = 24$, respectively. For slopeOP we get $1.95$ and $2.04$ with the channel pruning option. We plot the error bar using the $95\%$ empirical confidence interval.}
\label{fig:time}
\end{figure}

\subsection{A Minimum Angle Between Consecutive Segments}
\label{subsec:minangle}

We generate time series in scenario $1$ of the hat-shaped signal with an heavy-tailed noise: a Student distribution with a degree of freedom equal to $3$ with $\sigma = 24$. Simple computations give an angle of about $153^o$ between the two segments. We choose to fix a minimal angle between segments to $130^o$. We explore the value of the MSE returned by the standard slopeOP algorithm compared with the same algorithm with smoothing option for a range of penalty values. We also return the number of inferred segments. Each data-point is the mean over $100$ simulations.

Results presented in Figure~\ref{fig:smoothing} show that the MSE as well as the number of segments is lower with the min-angle option, whatever the penalty value. With a misspecified model or in presence of outliers, the minimum angle option (when we expect no small angle) improves the result of slopeOP. The penalty value has a reduced impact on the algorithm, which can be important for applications when the time series to analyse does not have a Gaussian noise structure.

\begin{figure}[!h]
\center
\includegraphics[width=1\textwidth]{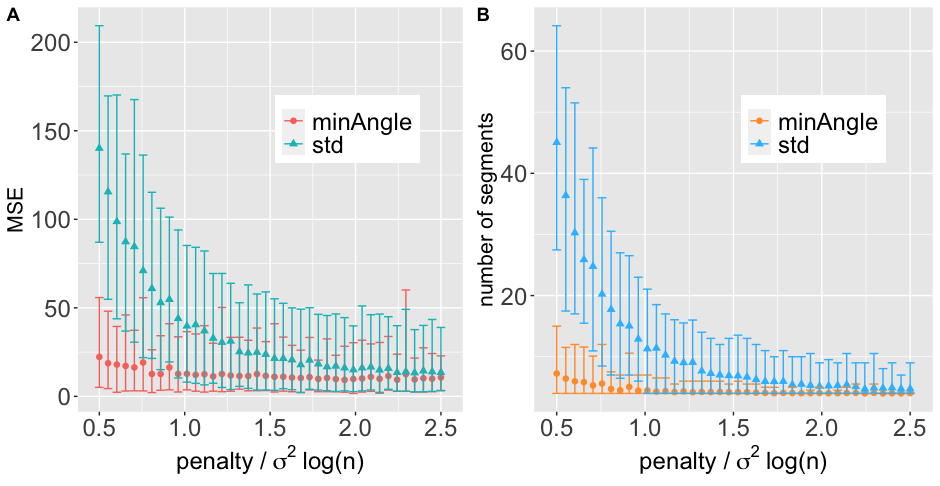} 
\caption{For a range of penalties from $\beta = 0.5 \sigma^2 \log(n)$ to $\beta = 2.5 \sigma^2 \log(n)$, $100$ one-hat-shaped time series have been simulated with $n=500$ for each penalty value and we ran slopeOP with a smoothing option ({\it minAngle = 130}) and without constraint  ({\it std}). On the left panel A, we compare the MSE between the inferred signal and the true one for the two methods; on the right panel B, we compare the number of inferred segments for the two methods. Results highlight the stability with respect to the penalty of the algorithm with minimal angle option.}
\label{fig:smoothing} 
\end{figure}

\section{Application to Antibiogram Image Analysis}
\label{sec:app}

Antibiotic Susceptibility Testing (AST) is a microbiology test used to guide antibiotic prescription in bacterial diseases by determining the susceptibility of bacteria to different antibiotics.

In disk diffusion AST, or antibiogram, cellulose disks impregnated with specific antibiotics are placed on the surface of Petri-dishes previously inoculated with a microorganism. Agar plates are incubated so that bacteria can grow everywhere, except around the cellulose disks that contain antibiotics to which the bacteria are susceptible. Then the diameter of the zone of non-growth (inhibition) surrounding each antibiotic disk is measured and compared to known minimum diameters to determine the susceptibility (the comparison is done with a sensibility of one millimeter)\cite{Jorgensen_disk_diffusion}.
Diameters can be read by eye with a ruler, but several automatic reading systems exist which process pictures of incubated plates \cite{Osiris2001,sirscan2013}.

The inhibition zone boundary is usually clear and easy to measure, but it can sometimes be hazy and fuzzy. In most cases, zone edges should be read at the point of complete inhibition\cite{EUCAST_protocol}, but determining this point can be challenging.

Several image processing procedures have been proposed for measuring inhibition diameters\cite{Hejblum2396,Gavoille1993,Costa2015}, but none of them focuses specifically on the issue of fuzzy borders. We propose the use of slopeOP to measure inhibition diameters in these cases. To our knowledge, a MSE-minimizing optimal segmentation algorithm has never been used to measure AST's inhibition dimeters.

We designed a processing pipeline that uses slopeOP to read inhibition diameters from the picture of an AST and tested it on one hundred inhibition zones presenting fuzzy borders. Fifty inhibition zones were taken from standard Mueller-Hinton growth medium antibiograms, the other fifty form blood agar antibiograms, which are darker and less contrasted.
\begin{figure}
    \centering
    \includegraphics[width=1\textwidth]{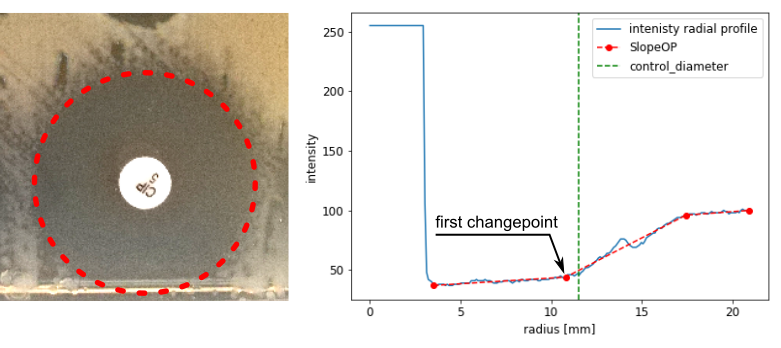}
    \caption{Picture of a part of an antibiogram showing an antibiotic disk and its inhibition zone (left). On the right a plot of the intensity radial profile and the segmentation with slopeOP with isotonic constraint.}
    \label{fig:AST_radial_profile}
\end{figure}
After normalizing the intensity of the AST picture, a sub-region of each inhibition zone was selected and centered on the the antibiotic disk. Then the radial intensity profile $I(r)$ was extracted:
we recorded $I$, the average intensity value of the 10 most intense pixels (the scale of the images is $\simeq10\frac{pixel}{mm}$) lying at fixed distance $r$ from the sub-region center (see figure \ref{fig:AST_radial_profile}). This method is a simple way to detect bacteria even if the inhibition zone in not a perfect disk. The average profile size in the data-set is 230 point.

slopeOP was used to segment the radial profile and determine the inhibition zone radius.
For the segmentation with slopeOP we take the signal starting at 3.5mm, i.e. just after the plateau corresponding to the pellet disk (from this point, the ideal signal is supposed to be isotonic).

slopeOP is used with
$\mathcal{S} = \{I(r)_{min}, \ldots, I(r)_{max}\}$ and a penalty $\beta=255$. We chose a high penalty in order to decrease the chance of detecting any changepoint before the bacterial halo (because of small imperfections in the sample or picture). As a comparison, we report in the appendix the results obtained with $\beta=2\sigma^2 \mathrm{log}(n)$ (cfr \ref{sup:ast}).

We tested both with and without the isotonic constraint (as the observed signal is supposed to be isotonic).

After the segmentation, we measure the inhibition radius $r_{inib}$ as the distance corresponding to the first change-point. The inhibition diameter is obtained as $d_{inhib}=2 \times r_{inhib}$.

As a comparison, we measured the diameters with the method suggested by Gavoille et al.\cite{Gavoille1993} which uses a student t-test and considers both the intensity and the texture of the pixels around the antibiotic disk.

The so measured diameters were compared with manual (by eye) measurements, by calculating the absolute difference.

\begin{table*}\small\centering
\begin{tabular}{c l rrr}
\toprule
&&\multicolumn{3}{c}{\bf diameter diff quantiles}\\
\bf AST type & \bf algorithm & 25\% & 50\% & 75\%\\
\midrule
M-H	    &CPOP	    &1.60	    &3.60	    &6.00\\
	    &SlopeOP	&0.29	    &0.91       &2.92\\
	    &SlopeOPi	&\bf0.24	&\bf0.69	&\bf1.71\\
	    &Ttest	    &0.50	    &1.50	    &3.25\\
\\
blood	&deltaCPOP	    &1.80   	&3.20	    &5.55\\
	    &deltaSlopeOP	&0.72	    &1.72	    &3.00\\
	    &deltaSlopeOPi	&\bf0.47	&\bf1.28	&\bf2.38\\
	    &deltaTtest	    &2.00   	&9.50   	&16.75\\
\bottomrule
\end{tabular}
\caption{Agreement between auto and manual reading. We observe the distribution quantiles of $\Delta{d} = |d_{a} - d_{c}|$ the absolute difference between the automatically measured diameter $d_a$ and the control value $d_c$. The smallest values are embolded.
}
\label{tab:AST_results}
\end{table*}

\begin{figure}[h]
\begin{center}
\includegraphics[width=6cm]{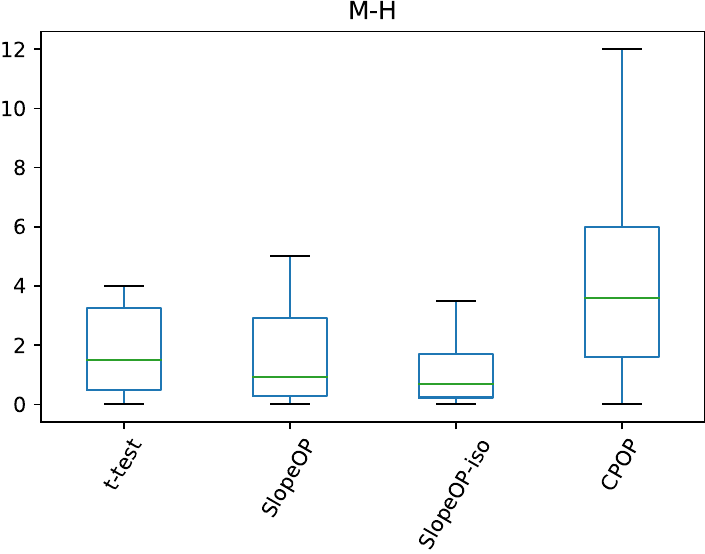}
\includegraphics[width=6cm]{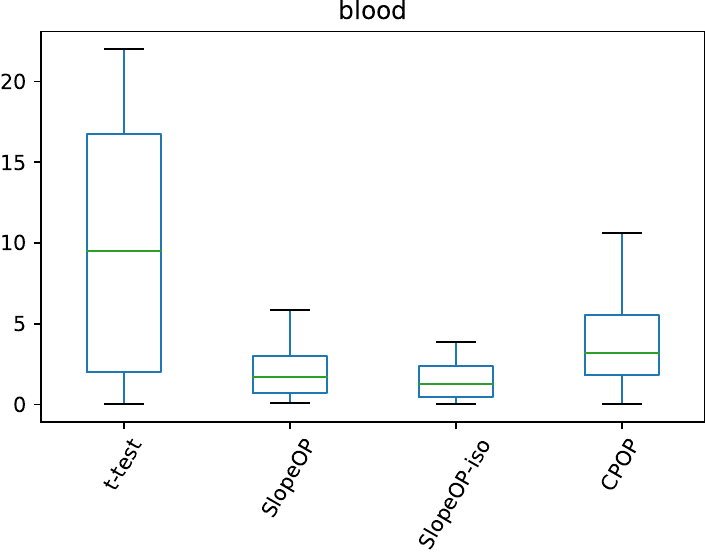}
\end{center}
\caption{Distribution of absolute diameter differences across the tested algorithms. The isotonic constraint visibly improves the accuracy of the measurement.}
\label{fig:AST_results}
\end{figure}

The results (Table \ref{tab:AST_results} and Figure \ref{fig:AST_results}) show a good consistency with the manual measurements. In the case of Mueller-Hinton antibiorgams, half of the diameter differences between our procedure and control are below the test sensibility (1mm) and 75\% are smaller than 1.5mm, which is comparable with the acknowledged inter-operator variability due to the subjectivity of the measurement ($\pm$ 1mm) \cite{Hombach3864}. The results on the blood agar ASTs are slightly worse because of the lower contrast in the images and a consequent decreased signal-noise ratio. In both cases the isotonic constraint yields improved results. 

\begin{table*}\small\centering
\begin{tabular}{lr rrr rr}
\toprule
&&\multicolumn{3}{l}{\bf diameter diff quantiles}&\bf speed&\bf  Changepoints\\
&&25\%&50\%&75\%&\\
\bf agar&\bf state density&&&&\\
\midrule
M-H     &1  &0.24&0.69&1.71&1.00&4.25\\
        &1/2&0.19&0.64&1.84&2.97&4.33\\
        &1/4&0.21&0.77&1.55&7.74&4.24\\
        &1/8&0.38&0.83&2.57&17.92&4.25\\
blood   &1&0.47&1.28&2.38&1.00&3.92\\
        &1/2&0.43&1.04&1.84&2.76&3.90\\
        &1/4&0.70&1.52&2.34&6.69&3.94\\
        &1/8&1.07&2.05&3.83&14.53&3.96\\
\bottomrule
\end{tabular}
\caption{Performance of slopeOP-isotonic at various states density. The calculation speed is reported as the inverse of the average execution time (normalized to density=1).
Changepoints=average number of changepoints found.
Mean execution time for slopeOP-iso1 is 5ms @ 2,3 GHz Intel Core i5.}

\label{tab:AST_results_states}
\end{table*}

Successively we tested the performance of slopeOP with isotonic constraint when reducing the density of states by 2,4 and 8. We defined the states space as
$\mathcal{S} = \{ I_{min}, I_{min}+s, I_{min}+2s, \ldots,I_{max}\}$ and repeated the measurements on the data-set at each value of $s=2,4,8$. The results are reported in Table \ref{tab:AST_results_states} and show a neat improvement in the execution speed by reducing the states density. Although the precision of the measurement decreases with decreasing states density, the measured diameter are still in reasonable accord with the control values.


\section{Conclusion}
\label{sec:conclusion}
Multiple change-point detection for change in slope with continuity constraint is a challenging problem in time-series analysis. The dynamic programming approach developed in CPOP \cite{fearnhead2018detecting} exactly optimize the penalized likelihood but suffers from a quadratic-to-cubic time complexity in data length. We proposed a novel approach with quadratic complexity restricting the set of possible slope and intercept values to a finite set. In addition, this new problem allows us to introduce constraints on successive segments. State values for the spatial discretization has to be chosen for each time-series and leads to two difficulties: how close are we to the continuous solution of CPOP? How to make a appropriate choice of state values? These questions remains widely open and could be a subject for further developments. A natural extension of this work would consider a dynamic programming algorithm following at each iteration a two-continuous parameter function (instead of one as in CPOP). In this parametrization the challenging question of pruning rules is yet open.

\section*{Acknowledgement}
We thank Guillem Rigaill (LaMME Evry, IPS2 Paris-Saclay) for joint mentoring of our intern (co-author Nicolas Deschamps de Boishebert), for the careful reading of this work and the many useful feedbacks.

\appendix

\section{Proof of Proposition \ref{prop:updateRule}}
\label{app:update}
The update rule is obtained by direct computations (we write ``vectors $\tau$ and $s$" to mean the constraints in~(\ref{optimSlope})): 
$$Q_t(v) = \min_{\substack{\text{vectors } \tau \text{ and } s \\ \tau_{k+1} = t\,,\, s_{k+1} = v}}\sum_{i=0}^k \big\{ \mathcal{C}(y_{(\tau_i+1):\tau_{i+1}}, s_i, s_{i+1}) + \beta\big\} - \beta$$
$$ = \min_{\substack{\text{vectors } \tau \text{ and } s}}\sum_{i=0}^{k-1} \big\{ \mathcal{C}(y_{(\tau_i+1):\tau_{i+1}}, s_i, s_{i+1}) + \beta\big\} - \beta + \mathcal{C}(y_{(\tau_{k}+1):t}, s_{k}, v) + \beta$$
$$ =\!\!\!\!\!\! \min_{\substack{0 \le t' < t\\ s_{min} \le u \le s_{max}}}\!\!\!\min_{\substack{\text{vectors } \tau \text{ and } s \\ \tau_{k} = t'\,,\, s_{k} = u}}\sum_{i=0}^{k-1} \big\{ \mathcal{C}(y_{(\tau_i+1):\tau_{i+1}}, s_i, s_{i+1}) + \beta\big\} - \beta + \mathcal{C}(y_{(t'+1):t}, u, v) + \beta$$
$$ = \min_{\substack{0 \le t' < t\\ s_{min} \le u \le s_{max}}}\left(Q_{t'}(u) + \mathcal{C}(y_{(t'+1):t}, u, v) + \beta\right)\,.$$

\section{Proof of Proposition \ref{P3}}
\label{app:pelt}

In order to prune we need to force the inequality
$$\mathcal{C}(y_{{(t'+1)}:T},u,v) \ge \mathcal{C}(y_{{(t'+1)}:t},u,v) + \mathcal{C}(y_{{(t+1)}:T},v,v)\,,$$
to be true for all $T \ge t+1$ and fixed $t',t,u,v$. Using the cost definition (\ref{costLinear}), the inequality is equivalent to
$$(v-u) \Bigg(\sum_{i=t'+1}^{t}y_i \frac{(i-t')(T-t)}{(t-t')(T-t')} + \sum_{i=t+1}^{T-1}y_i \frac{T-i}{T-t'} \Bigg)$$ 
$$\quad\quad\quad\quad\quad\quad\quad\ge (v-u)(T-t) \Bigg(\frac{u + 2v}{6} + \frac{v-u}{12(t-t')(T-t')}\Bigg)$$
and then
\begin{equation}
\label{pruningProof}
(v-u)g_t(T) \ge  (v-u)\Bigg((T-t') \Big(\frac{u + 2v}{6}\Big) + \frac{v-u}{12(t-t')} -\frac{S_t^{t'}}{t-t'} \Bigg)
\end{equation}
with
$$ g_t(T) = \left\{
\begin{array}{ll}
\displaystyle \sum_{i=t+1}^{T-1}y_i \frac{T-i}{T-t} \quad &\hbox{ for }  T=t+2,\dots,n\,, \\
0 &\hbox{ for }  T=t+1 \,.
\end{array}
\right.
$$
As the right hand side of inequality (\ref{pruningProof}) is linear in $T$ we choose the following strategy: find coefficients $(\alpha_t^-,\alpha_t^+,\gamma^-_t,\gamma^+_t)$ of an upper and lower linear approximation such that
$$\alpha_t^+ T + \gamma^+_t \le g_t(T) \le \alpha_t^- T + \gamma^-_t\,, \quad T= t+1,\dots,n\,.$$
We then introduce the linear-in-$T$ functions $f^-$ and $f^+$ defined in the Proposition. To prove that the inequality $f^+(T) \ge 0$ holds for all $T=t+1$ to $n$ we only need to have $f^+(t+1) \ge 0$ and $f^+(n) \ge 0$ due to the linearity of $f^+$. With the same argument for $f^-$ the result is proven.

\section{Variance estimation: complement with $10$ segments}
\label{app:variance}
With smaller segments as presented in Figure \ref{fig:sigmaEstim2}, the estimator "HALL diff" remains the more efficient but overestimates the variance in case of a sinusoidal signal, in particular for lower noise levels. Nevertheless, this estimation remains less biased than the two others.

\begin{figure}[!h]
\center\includegraphics[width=0.9\textwidth]{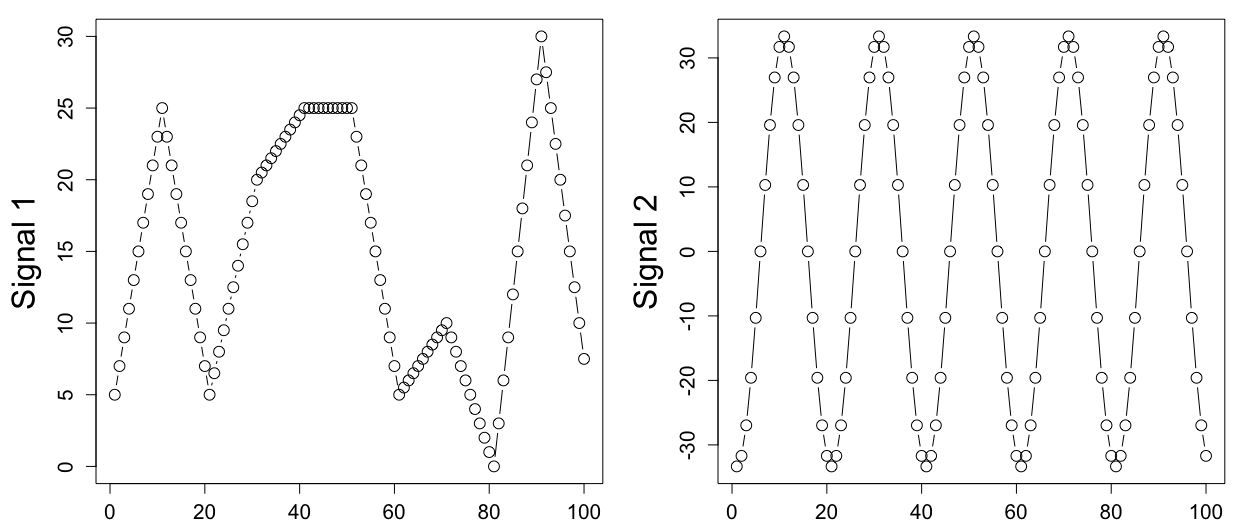} 
\caption{Linear and sinusoidal signals of length $100$ with $10$ segments.}
\label{fig:sigmaEstim2}
\end{figure}
\begin{table}[!h]
\begin{tabular}{c|c|cccccc}
method&signal& $\sigma$ &1& 2 & 3 & 4 & 5 \\
\hline\hline
MAD&slope&mean($\hat\sigma$)& 1.71 & 2.37 & 3.26 & 4.19 & 5.163   \\
&(signal 1) &sd($\hat\sigma$)& 0.15 & 0.26 & 0.39 & 0.52 & 0.65 \\
\cline{2-8}
&sinus&mean($\hat\sigma$)& 7.60 & 7.30 & 7.21 & 7.37 & 7.75    \\
&(signal 2)&sd($\hat\sigma$)& 0.26 & 0.35 & 0.45 & 0.56 & 0.67
  \\
\hline\hline
HALL&slope&mean($\hat\sigma$)& 2.71 & 3.21 & 3.92 & 4.72 & 5.59  \\
 &(signal 1)&sd($\hat\sigma$)& \it 0.056 & \it 0.13 & \it 0.2 & \it 0.28 & \it 0.37  \\
\cline{2-8}
&sinus&mean($\hat\sigma$)& 11.2 & 11.3 & 11.5 & 11.8 & 12.2  \\
&(signal 2)&sd($\hat\sigma$)& 0.048 & \it 0.097 & \it 0.15 & \it 0.21 & \it 0.27\\
\hline\hline
HALL Diff&slope&mean($\hat\sigma$)& \bf 1.18 & \bf 2.09 & \bf 3.06 & \bf 4.04 & \bf 5.02  \\
 &(signal 1)&sd($\hat\sigma$) & 0.085 & 0.17 & 0.26 & 0.35 & 0.44 \\
\cline{2-8}
&sinus&mean($\hat\sigma$) & \bf 2.43 & \bf 2.98 & \bf 3.72 & \bf 4.56 & \bf 5.46  \\
&(signal 2)&sd($\hat\sigma$)& \it 0.039 & 0.12 & 0.22 & 0.31 & 0.40  \\
\hline\hline
\end{tabular}
\caption{Variance estimation with MAD, HALL and HALL Diff estimators based on the two signals presented in Figure \ref{fig:sigmaEstim2}. The closest values to the true sigma are in bold; the smallest standard deviations in italic. We simulated $10^4$ time-series for each experiment. We notice that all estimators struggle to return an unbiased estimation with the sinusoidal signal, in particular with the MAD and HALL estimators for which the estimation seems no more sigma dependent.}
\label{tab:variance2}
\end{table}

\newpage

\section{Simulation Results for the $4$ Scenarios}
\label{app:scenarios}

\begin{figure}[!h]
\center
\includegraphics[width=\textwidth]{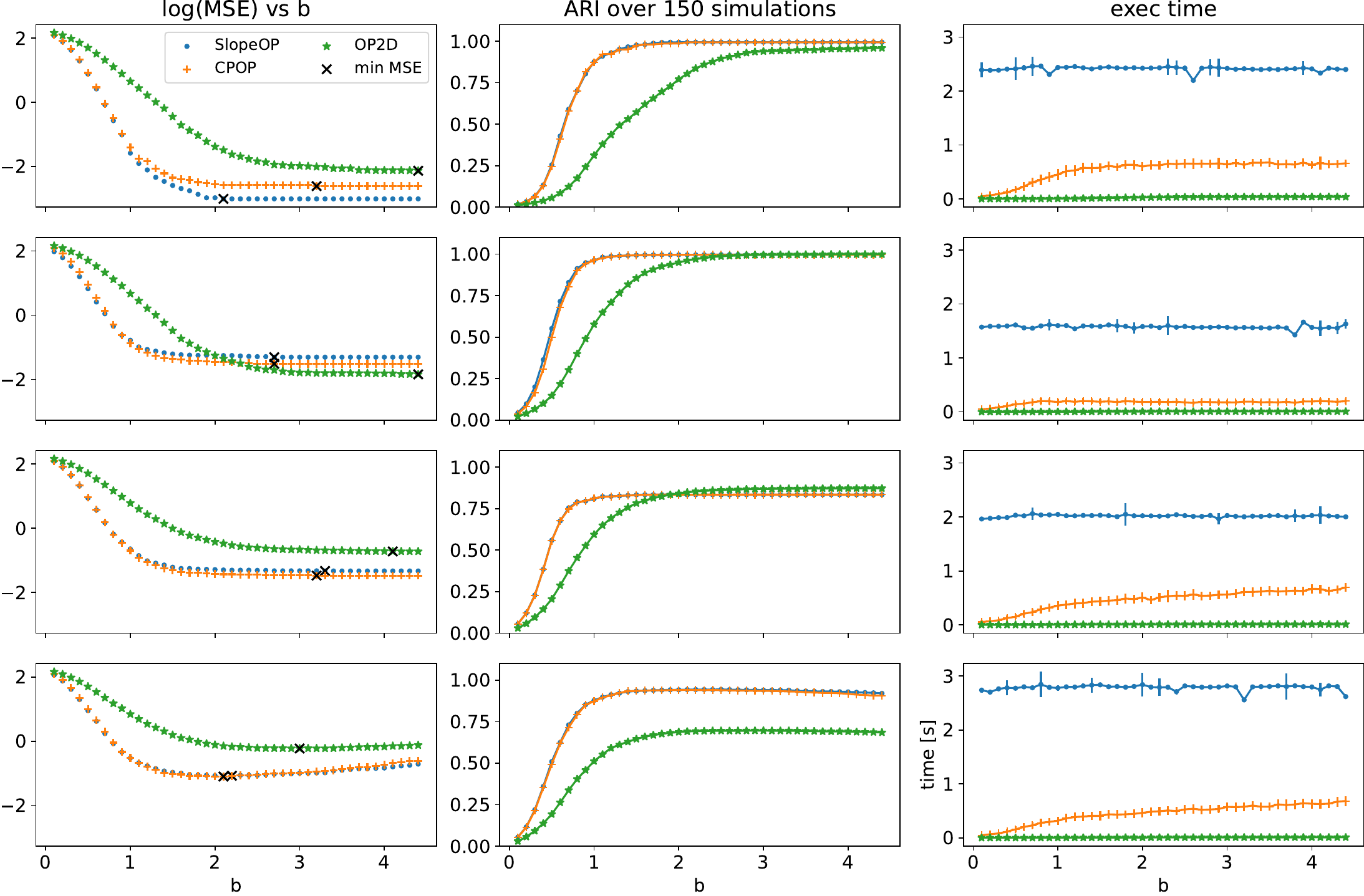} 
\caption{Results for scenarios $1$ to $4$ with noise $\sigma = 3$}
\label{scenarioResults3}
\end{figure}

\begin{figure}[!h]
\center
\includegraphics[width=\textwidth]{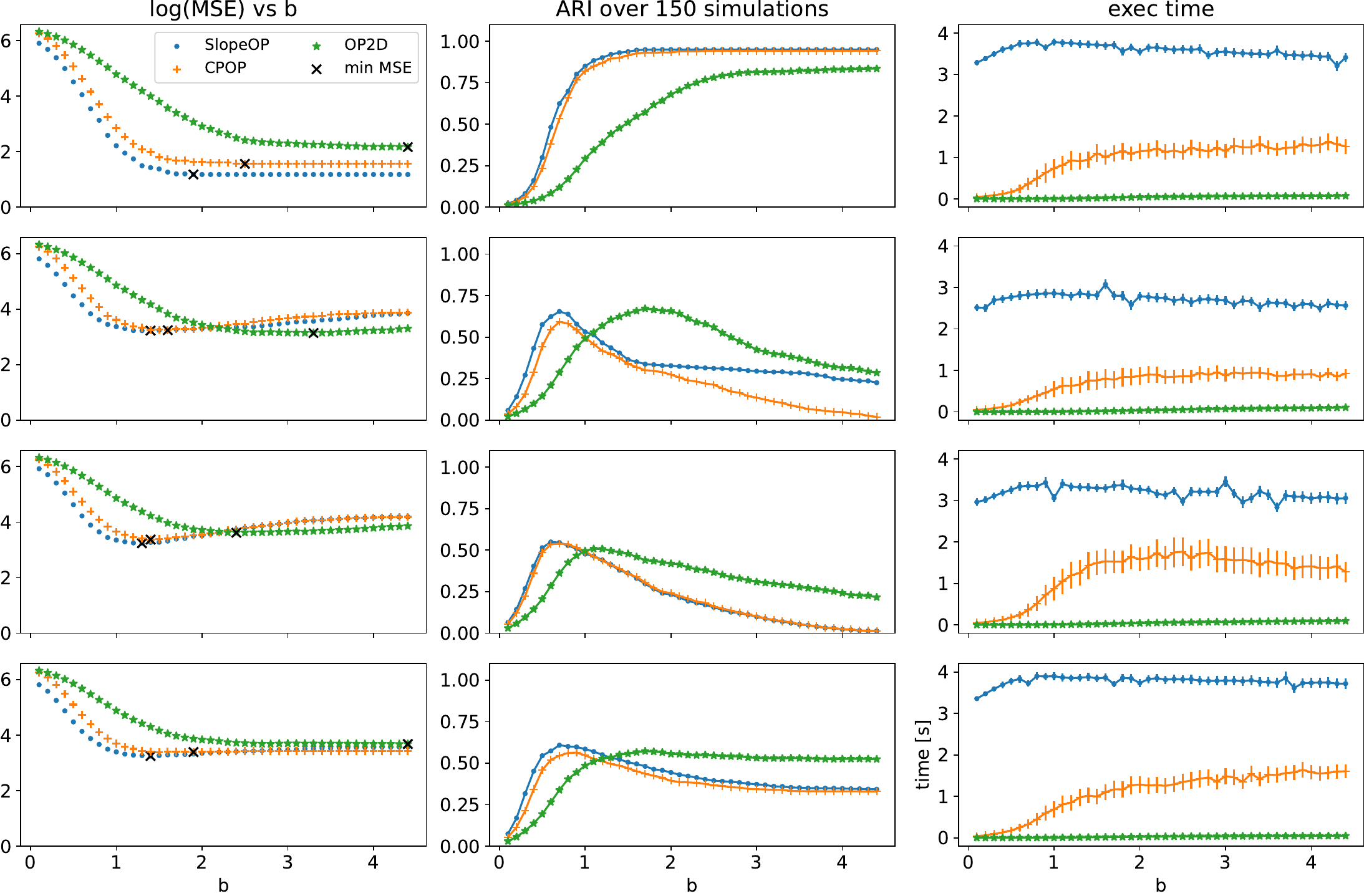} 
\caption{Results for scenarios $1$ to $4$ with noise $\sigma = 24$}
\label{scenarioResults24}
\end{figure}

\newpage

\section{Time complexity with $10$-hat-shaped signals}
\label{app:time}

\begin{figure}[!h]
\center
\includegraphics[width=1\textwidth]{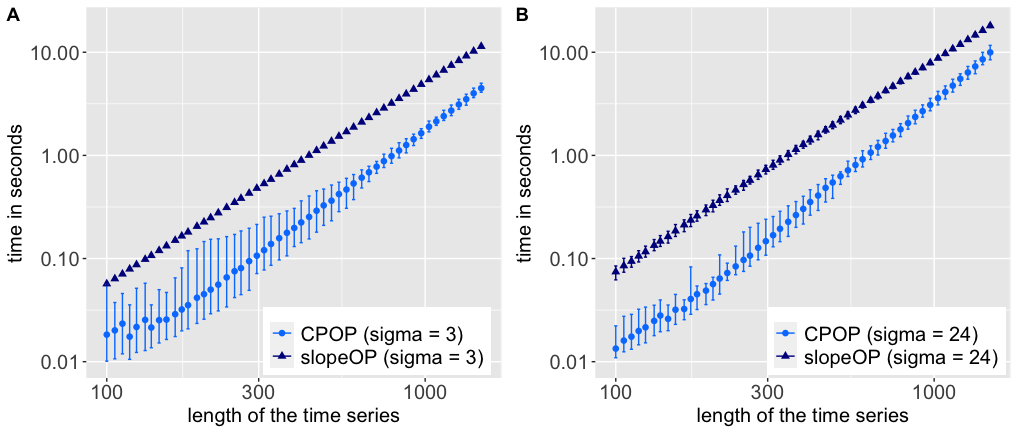} 
\caption{Computational time of CPOP versus slopeOP for $10$-hat-shaped data with noise level $\sigma =3$ and $\sigma =24$ in log-log scale. The coefficient $q$ in complexity $O(n^q)$ is equal to $2.16$ and $2.47$ for CPOP for $\sigma = 3$ and $\sigma = 24$, respectively. For slopeOP we get $1.96$ and $2.03$ with the channel pruning option.}
\label{fig:time10hats}
\end{figure}

\newpage

\section{AST application - complementary results}
\label{sup:ast}
\begin{table*}\small\centering
\begin{tabular}{llrrr}
\toprule
	&	&\multicolumn{3}{l}{diameter diff. quantiles}\\
	&	&25\%	&50\%	&75\%\\
AST type	&algorithm	&	&	&\\
\midrule
blood	&deltaCPOP	    &8.00	    &12.50	    &18.00\\
	    &deltaSlopeOP	&8.35   	&12.10	    &15.28\\
	    &deltaSlopeOPi	&\bf1.63	&\bf3.44	&\bf7.23\\
	    &deltaTtest	    &2.00   	&9.50	    &16.75\\
\\
M-H	    &deltaCPOP	    &9.00	    &16.00	    &18.00\\
	    &deltaSlopeOP	&6.46	    &12.05	    &17.01\\
	    &deltaSlopeOPi	&1.89	    &8.49	    &15.56\\
	    &deltaTtest	    &\bf0.50	&\bf1.50    &\bf3.25\\
\bottomrule
\end{tabular}
\caption{Agreement of auto and manual reading.}
\label{tab:AST_results_low_beta}
\end{table*}

\begin{figure}[h]
\begin{center}
\includegraphics[width=5cm]{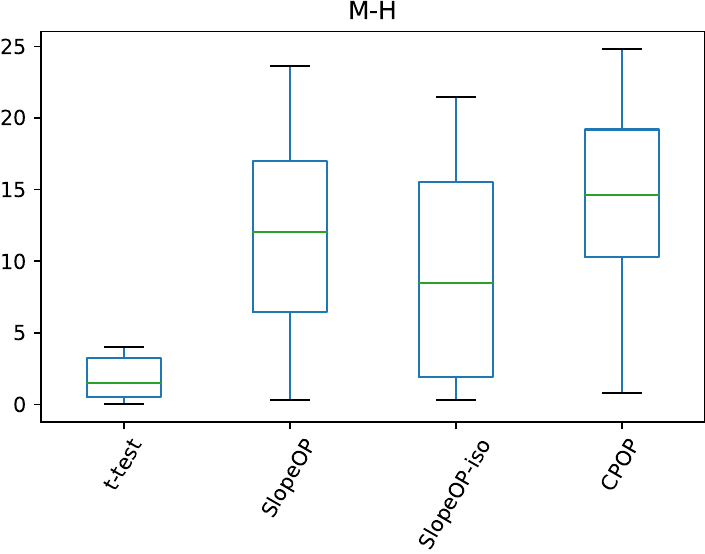}
\includegraphics[width=5cm]{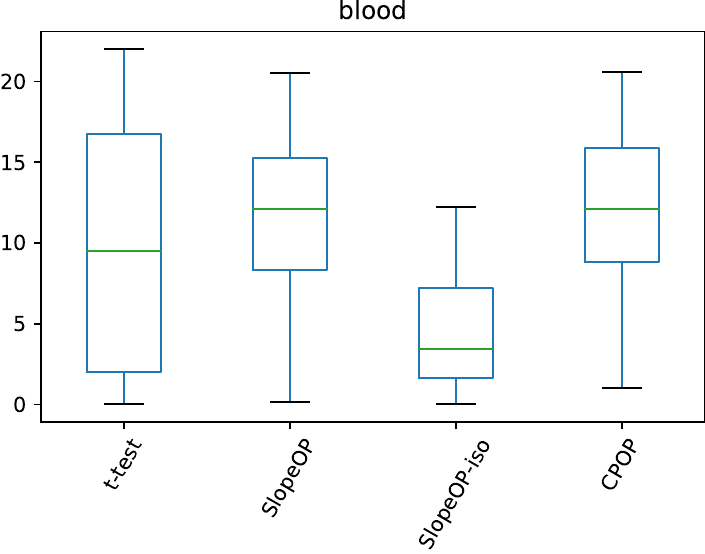}
\end{center}
\caption{Distribution of absolute diameter differences across the tested algorithms. For M-H antibiograms, Gavoille's t-test method shows the best results, whereas in the case of blood enriched culture medium, SlopeOP with isotonic contraint outperforms the others. Both SlopeOP unconstrained and CPOP give poor measurements, underlining the fact that the isotonic constraint improves the results}
\label{fig:AST_results_low_beta}
\end{figure}

We repeated the automatic experiment of Section \ref{sec:app} with a variable penalty calculated as $\beta=2\sigma^2 \mathrm{log}(n)$ where the variance $\sigma^2$ is estimated with the HallDiff estimator presented in this paper (cfr. \ref{sec:varest}).
The average penalty value obtained is $4$ which results in the detection of more changepoints in the signal (compared to Section \ref{sec:app}. In this case, the assumption that the first detected changepoint is close the bacterial boundary is not appropriate and results in much worse measurement for all tested segmentation algorithms.
Nevertheless, the results (Table \ref{tab:AST_results_low_beta} and Figure \ref{fig:AST_results_low_beta}) show an neat improvement in the measurement accuracy when using the isotonic constraint. In the case of Mueller-Hinton antibiotics, even if the isotonic constraint ameliorates the result of the segmentation, none of the proposed segmentation algorithms achieves good measurements compared to Gavoille's t-test method.
In the case of blood enriched culture medium instead, where the t-test method gives less accurate results, the isotonic constraint of SlopeOP shows the best measurement accuracy.

\begin{table*}\small\centering
\begin{tabular}{llrrrrr}
\toprule
	&	&\multicolumn{3}{l}{abs. diameter difference}	&speed	&av. changepoints number\\
	&	&25\%	&50\%	&\multicolumn{3}{l}{75\%}\\
agar	&states density	&	&	&	&	&\\
\midrule
blood	&1  	&1.63	&3.44	&7.23	&1.00	&7.24\\
	    &1/2	&0.93	&4.43	&7.87	&2.87	&5.60\\
	    &14	    &2.03	&4.20	&6.46	&7.41	&8.76\\
	    &1/8	&1.76	&3.04	&6.97	&17.02	&8.14\\
mh	    &1	    &1.89	&8.49	&15.56	&0.43	&9.76\\
	    &1/2	&2.54	&8.65	&15.19	&1.00	&8.29\\
	    &14	    &1.12	&5.00	&11.15	&2.80	&10.98\\
	    &1/8	&1.94	&3.69	&10.62	&6.96	&10.82\\
\bottomrule
\end{tabular}
\caption{Performance of slopeOP-isotonic at various states density.
The calculation speed is reported as the inverse of the average execution time (normalized to density=1).
Changepoints=average number of changepoints found.
NOTE: Mean execution time for slopeOP-iso1 is 5ms @ 2,3 GHz Intel Core i5.}
\label{tab:AST_results_states_1_low_beta}
\end{table*}

\bibliographystyle{abbrv}
\bibliography{Bibliography}

\end{document}